\documentclass{article}

\usepackage{PRIMEarxiv}

\RequirePackage{amsthm,amsmath,amsfonts,amssymb}

\usepackage[numbers]{natbib}
\usepackage[utf8]{inputenc} % allow utf-8 input
\usepackage[T1]{fontenc}    % use 8-bit T1 fonts
\usepackage{url}            % simple URL typesetting
\usepackage{booktabs}       % professional-quality tables
\usepackage{amsfonts}       % blackboard math symbols
\usepackage{nicefrac}       % compact symbols for 1/2, etc.
\usepackage{microtype}      % microtypography
\usepackage{lipsum}
\usepackage{fancyhdr}       % header
\usepackage{graphicx}       % graphics
\graphicspath{{media/}}     % organize your images and other figures under media/ folder
\RequirePackage[OT1]{fontenc}
\RequirePackage{amsthm,amsmath}
\usepackage{hyperref}
\hypersetup{colorlinks,citecolor=blue,urlcolor=blue}
\RequirePackage{threeparttable}
\RequirePackage{graphicx}
\RequirePackage{subcaption}
\RequirePackage{caption}
\RequirePackage{epstopdf}
\RequirePackage{float}
\RequirePackage{soul} % wrap up the context in underline mode
\RequirePackage{longtable}
\RequirePackage{array}
\RequirePackage{bbm}
\usepackage{algorithm,algorithmic}% http://ctan.org/pkg/algorithms
\usepackage{amssymb}
\usepackage{amsthm}
\usepackage{booktabs}
\usepackage{verbatim}
\theoremstyle{plain}
\newtheorem{theorem}{Theorem}
\newtheorem{proposition}{Proposition}

\newtheorem{corollary}{Corollary}
\newtheorem{lemma}{Lemma}
\theoremstyle{remark}
\newtheorem*{remark}{Remark}

\newcommand{\s}{\boldsymbol{s}}
\newcommand{\iid}{ {\stackrel{i.i.d.}{\sim}} }

\numberwithin{equation}{section}
\theoremstyle{plain}

\title{Inverse set estimation and inversion of simultaneous confidence intervals
}

\author{
  Junting Ren \\
  Division of Biostatistics\\
  University of California San Diego \\
  \texttt{j5ren@ucsd.edu} \\
   \And
  Fabian J.E. Telschow \\
  Institute of Mathematics \\
  Humboldt Universit\"{a}t zu Berlin \\
  \texttt{fabian.telschow@hu-berlin.de} \\
  \And
  Armin Schwartzman \\
  Division of Biostatistics and Hal{\i}c{\i}o\u{g}lu Data Science Institute \\
  University of California San Diego \\
  \texttt{armins@ucsd.edu} \\
}

\begin{document}
\maketitle
\begin{abstract}
Motivated by the questions of risk assessment in climatology (temperature change in North America) and medicine (impact of statin usage and COVID-19 on hospitalized patients), we address the problem of estimating the set in the domain of a function whose image equals a predefined subset. Existing methods that construct confidence sets require strict assumptions. We generalize the estimation of such sets to dense and non-dense domains with protection against ``data peeking'' by proving that confidence sets of multiple levels can be simultaneously constructed with the desired confidence non-asymptotically through inverting simultaneous confidence bands. A non-parametric bootstrap algorithm and code are provided.
\end{abstract}

% keywords can be removed

\section{Introduction}
\subsection{Motivation}
%Our motivation comes is related to the potential impact of global temperature changes expected in the upcoming century. There is a crucial need to evaluate which areas of the world could be especially vulnerable to extreme shifts in temperature.
One motivating problem for our work comes from the data analysis in \cite{sommerfeld2018confidence}. The data used here is obtained from the North American Regional Climate Change Assessment Program (NARCCAP) project \cite{mearns2013climate}. This data comprises two sets of 29 geographically registered arrays of average seasonal temperatures for summer (June-August) and winter (December-February), during two time frames: the late 20th century (1971–1999) and the mid-21st century (2041-2069). The aim of the analysis is to identify specific geographical regions where the difference in average summer or winter temperatures between these two periods exceeds a certain benchmark, with the intention of helping policymakers focus on regions that are at higher risk for effects of climate change.

Mathematically, the regions with temperature difference exceeding $c$ degrees can be defined as $\mu^{-1}(U)=\{\s \in \mathcal{S}:\mu(\s) \in U \}$: the set in the closed domain $\mathcal{S}$ such that the function output $\mu(\s)$ (true difference in temperature) is in the half interval $U=[c,\infty)$. We call $\mu^{-1}(U)$ {\em inverse set} because it is the preimage or inverse image of a set $U \subset \mathbb{R}$ under a deterministic function $\mu: \mathcal{S} \mapsto \mathbb{R}$. Suppose $\mu$ is unknown but data is available to construct an estimator $\hat\mu_n$ where $n$ is the sample size. A point estimate of the inverse set $\mu^{-1}(U)$ can be constructed as $\hat\mu_{n}^{-1}(U)$, indicated by the inside of the green contours in the middle lower and upper panels of Figure \ref{fig:climate_result} for $U=[2,\infty)$. But how do we assess the spatial uncertainty of such an estimate?

To assess this uncertainty, \cite{sommerfeld2018confidence} introduced Coverage Probability Excursion (CoPE) sets, here called inner and outer confidence sets (CSs), that are sub- and super-sets of the target inverse set, i.e.
\begin{equation*}
    \mathrm{CS}_\mathrm{in}(U) \subseteq \mu^{-1}(U) \subseteq \mathrm{CS}_{\mathrm{out}}(U). \label{eq:intro_inverse_set_def}
\end{equation*}
with a certain pre-specified probability, say 95\%. We call these CSs due to their analogy to confidence intervals, with the ``lower bound'' being $\mathrm{CS}_\mathrm{in}(U)$ and the ``upper bound'' being $\mathrm{CS}_{\mathrm{out}}(U)$. In the NARCCAP application, for $U=[2,\infty)$, $\mathrm{CS}_{\mathrm{in}}(U)$ and $\mathrm{CS}_{\mathrm{out}}(U)$ are indicated respectively by the inside of the red and blue contours in the middle lower panel of Figure \ref{fig:climate_result}.

In order to have a precise control of $\mathbb{P}\left(\mathrm{CS}_\mathrm{in}(U) \subseteq \mu^{-1}(U) \subseteq \mathrm{CS}_{\mathrm{out}}(U)\right)$, \cite{sommerfeld2018confidence} assume that the domain $\mathcal{S}$ is a dense subset of $\mathbb{R}^d$, that both $\mu(\s)$ and $\hat{\mu}_n(\s)$ are continuous whenever they have values close to $c=2$, and the desired coverage is only guaranteed asymptotically  as the sample size $n$ goes to infinity. These assumptions lead to a rather complicated proof and limit its applicability and generality.

Even for the NARCCAP climate change data that is used as main illustration in \cite{sommerfeld2018confidence}, the assumptions are not strictly satisfied. The data consists of only 29 samples per location, at which the algorithm in \cite{sommerfeld2018confidence} fails to construct CSs that achieve the correct coverage in simulations, as we show in Section \ref{sec:simulation}. In addition, the temperature data is only observed on a finite set of locations, so it is not strictly dense in $\mathbb{R}^2$ \cite{mearns2013climate}. 

\begin{figure}
    \centering
    \includegraphics[width=14cm]{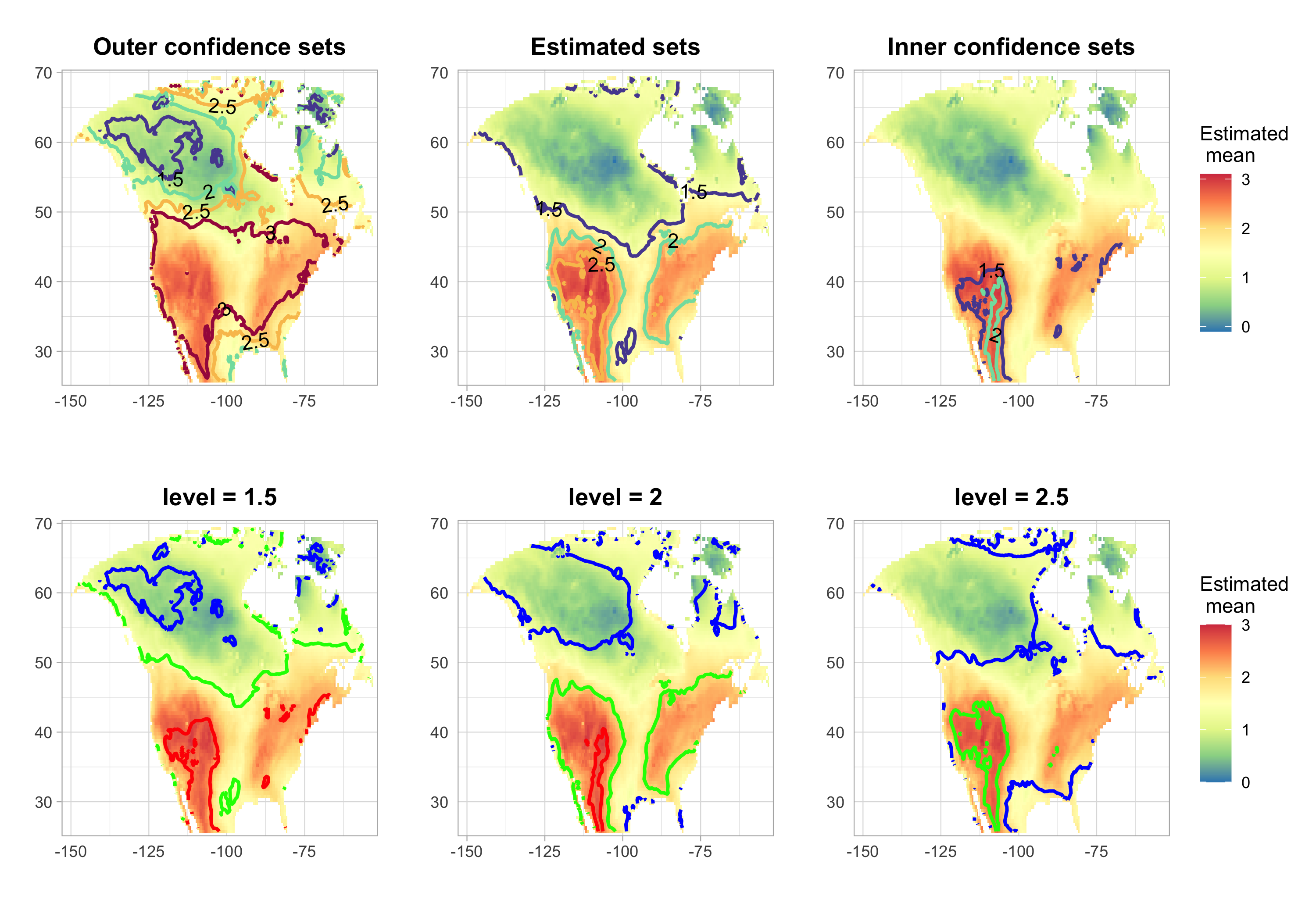}
    \caption{Confidence sets for the increase of the mean summer temperature (June–August) in North America between the 20th and 21st centuries according to the specific climate model analyzed in \cite{sommerfeld2018confidence}. Heat maps show the estimate of the mean difference. The first row displays the contours of the outer confidence sets, estimated inverse set, and the inner confidence sets, for various levels. The three plots in the second row display the confidence sets for the inverse sets, where the estimated mean difference is greater or equal to the individual level 1.5, 2.0, or 2.5 respectively. In the second row, the blue line is the contour of the outer confidence set, the green line is the contour of the estimated inverse set and the red line is the contour of the inner confidence set.}
    \label{fig:climate_result}
\end{figure}

Due to its required assumptions, the original approach \cite{sommerfeld2018confidence} is designed for dense functional data and cannot be applied to other data types such as multiple regression data. For instance, using the data in \cite{daniels2020relation}, consider the problem of identifying patient characteristics that lead to a risk of having a severe outcome which is higher than a certain threshold. Figure \ref{fig:covid_confidence_set} shows the estimated probability of hospitalized patients having a severe outcome, depending on age, COVID status, and statins medication status, obtained using multiple logistic regression. The use of statin [adjusted odds ratio (aOR) 0.78, confidence interval (CI) 0.66 to 0.93] is associated with decreased probability of severe outcome. Using CSs, we can visualize the protective effect of statin for better interpretation as detailed in Section \ref{sec:wide_application}. However, since categorical covariates are discrete, the domain is not a dense subset of $\mathbb{R}^d$ where $d$ is the number of covariates. Therefore, the original method or other existing methods for constructing CSs are not applicable in this scenario, but the method we propose here is.

In terms of statistical inference, the existing approaches require the investigator to set a fixed excursion threshold level, for example $2^{\circ}$C in the climate change data. This threshold depends on the context. Yet, setting a good threshold is difficult even for domain experts \cite{rogelj2009halfway}. Why is $2^{\circ}$C important but not $1.5^{\circ}$C? It is natural, and almost unavoidable, for investigators to try different thresholds and choose those that give most meaningful results. An analysis example using multiple thresholds is shown in the upper panel of Figure \ref{fig:climate_result} or in Figure \ref{fig:covid_confidence_set}. Therefore, to assure valid inference with control of type I error rate, the coverage of the CSs should be {\em simultaneous} over all thresholds. 

\subsection{Contributions}
This paper proposes an elegant solution to overcome the limitations of the previous methods. The answer, it turns out, is to construct confidence sets by inverting pre-built simultaneous confidence intervals (SCIs) which are widely applicable in different data modalities. In this paper, we underscore the broad applicability of our method, primarily concentrating on the construction of CSs for two prevalent but distinct data modalities: dense functional data and multiple regression data. The performance of various algorithms in constructing SCIs for dense functional data has been rigorously evaluated in prior work \cite{telschow2022simultaneous}. For multiple regression data, although the non-parametric bootstrap algorithm has been validated as a method for capturing the asymptotic distribution of multiple linear regression coefficients \cite{freedman1981bootstrapping}, its efficacy in constructing SCIs within a finite sample setup remains largely unexplored. Consequently, we introduce a non-parametric bootstrap algorithm, supplemented by R code, for constructing SCIs in multiple regression and provide a comprehensive evaluation of its performance. Our simulation results reveal that this approach not only controls the predetermined Type I error rate effectively but also maintains robustness despite finite sample sizes and does not necessitate the continuity of covariates. Furthermore, our method of inverting pre-built SCIs ensures that the coverage probability of the confidence sets, for any given threshold $c \in \mathbb{R}$, aligns precisely with the SCI coverage rate, as corroborated by our theorems. Inspired by Goeman \cite{goeman2011multiple,goeman2021only}, this safeguards against ``data peeking'' in exploratory data analysis, thereby enabling researchers to construct confidence sets for any threshold $c$ without concerns about compromising the control over the Type I error rate. The algorithm, simulation, and data application code associated with our study is accessible online at \url{https://github.com/junting-ren/inverse_set_SCI}.

\begin{figure}
    \centering
    \includegraphics[width=14cm]{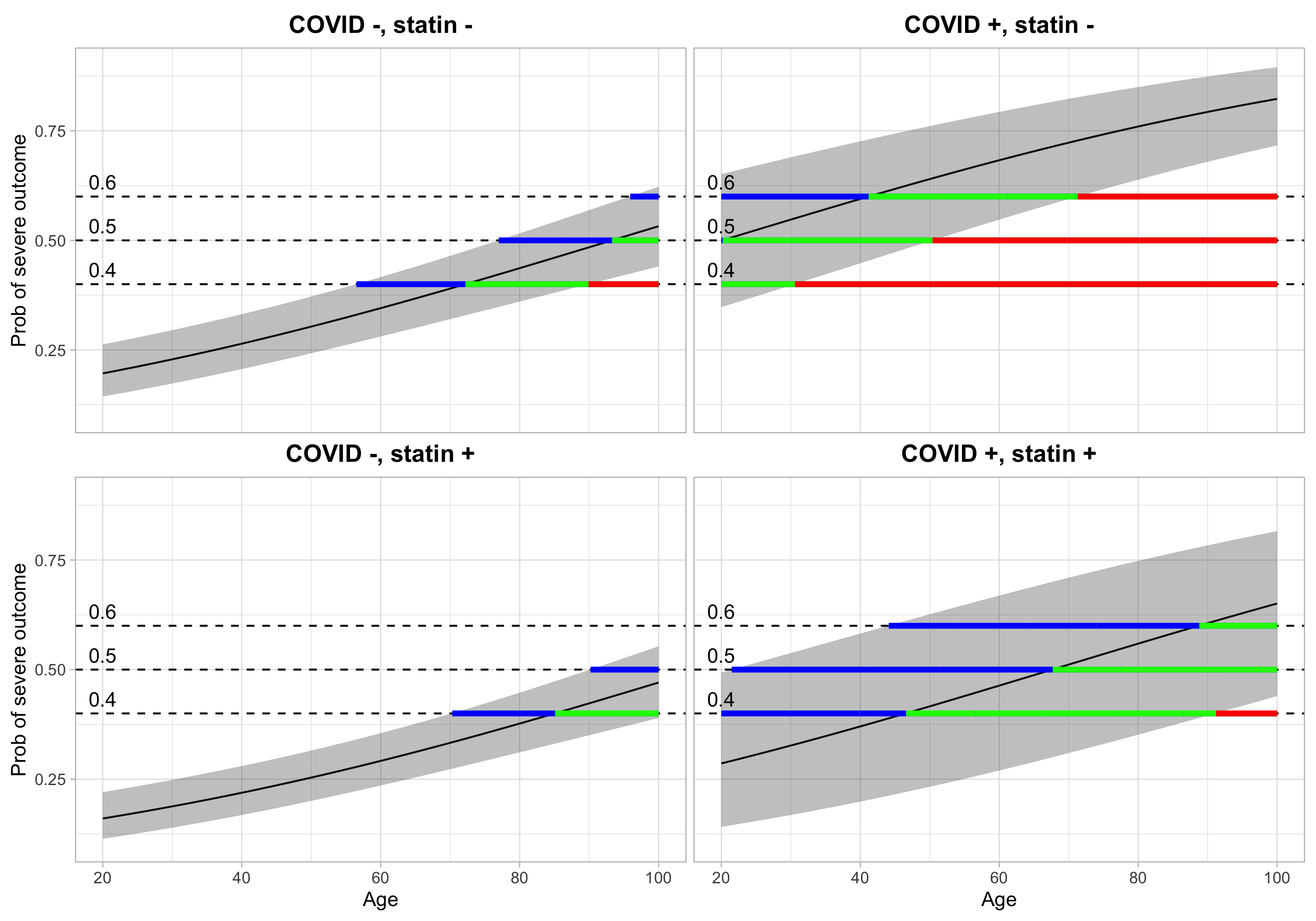}
    \caption{Simultaneous confidence set for the probability of severe outcome. We fixed other variables at ACE = 0, ARB = 0, sex = Male, CKD = 1, hypertension=1, CVD = 1, diabetes=1, obesity = 1. The gray shaded area is the 95\% SCIs, the solid black line is the estimated probability. The red horizontal line shows the inner confidence sets (where the lower SCIs are greater than the corresponding level) which are contained in the estimated inverse upper excursion set colored as the green and red horizontal line (where the estimated means are greater than the corresponding levels); the outer confidence sets are colored by the blue, green and red line (where the upper SCIs are greater than the corresponding levels) and contain both the estimated inverse sets and the inner confidence sets.}
    \label{fig:covid_confidence_set}
\end{figure}

\subsection{Other existing inverse set estimation methodology}
In addition to application to climate change \cite{sommerfeld2018confidence}, inverse set estimation methods are applied in many other different fields, such as astronomy \cite{jang2006nonparametric}, medical imaging \cite{willett2005level, bowring2019spatial, bowring2021confidence}, dose-effect finding \cite{jankowski2014random}, and geoscience  \cite{french2017assessing}. Furthermore, there is a growing trend to quantify the effect size for genomic regions rather than just testing the null hypothesis \cite{weinstein2013selection, benjamini2019selection}, where inverse set estimation methods can be utilized to quantify the uncertainty of genomic regions with effects greater than a certain threshold. 

However, just like the aforementioned method of \cite{sommerfeld2018confidence}, existing inverse set estimation methods are only applicable to specific kinds of data and require strict assumptions.
%For example, the method from \cite{sommerfeld2018confidence} mentioned in the previous section is proposed for dense functional data, which requires the underlying true mean function to be continuous, and the coverage probability is only asymptotic for large sample size.
Other methods are specifically designed for scenarios where the function $\mu$ is a density function \cite{mammen2013confidence, saavedra2016comparative, qiao2019nonparametric}. Inverse set methods have been also developed for stochastic processes (random functions), but they require the process itself to be Gaussian and data must be observed on a fixed grid \cite{french2013spatio,french2014confidence, bolin2015excursion}. The additional significant issue with all the inverse set estimation methods above is that the estimated confidence sets are only valid for a single threshold $c$, for example, estimating the set $\mu^{-1}[c, +\infty)$ for a fixed threshold $c$. 

\subsection{Existing simultaneous confidence interval methods}
Since the proposed inverse set estimation method is based on SCIs, it is worth reviewing existing SCI methods, to which our method would be applicable. For dense functional data, researchers constructed SCIs based on functional central limit theorems in the Banach space using Monte-Carlo simulations with an estimate of the limiting covariance structure \cite{degras2017simultaneous,degras2011simultaneous,cao2014simultaneous}, based on bootstrap \cite{wang2020simultaneous,chang2017simultaneous}, and based on the Gaussian Kinematic formula \cite{telschow2022simultaneous}. For sparse functional data,  SCIs are built using functional principal component analysis \cite{yao2005functional, ma2012simultaneous}. For high dimensional data such as genomics data with discrete indexing, valid SCIs are built for high dimensional but a finite number of parameters before selection \cite{park2007simultaneous} or after selection \cite{hwang2013empirical,qiu2007sharp}. For survival data, SCIs for survival functions are built using Greenwood's variance formula under large sample sizes \cite{nair1984confidence}, as well as SCIs for the difference or ratio of two survival functions \cite{parzen1997simultaneous, mckeague2002simultaneous}. For regression problems, researchers are often interested in how the response $y$ changes with a vector of predictors $\boldsymbol{x}$, or the magnitude of the regression coefficients. Therefore, SCIs can be constructed for $y$ on the range of $\boldsymbol{x}$ for simple linear regression \cite{liu2008construction} and multiple regression on the dense compact subset of continuous covariates in $\mathbb{R}^d$ \cite{sun1994simultaneous}. However, to the best of our knowledge, there is no practical bootstrap algorithm nor accessible code online that constructs SCIs for linear combinations of coefficients of multiple regression that is valid under finite sample size, which is addressed in the current paper with our algorithm and code.  

\subsection{Outline}
After stating and proving the main theorem and corollaries in Section \ref{sec:theory}, we present the results of simulation studies that validate our method for continuous domains using dense functional data and regression mean prediction on a fine grid of predictors. For discrete domains, confidence sets for regression coefficients are constructed using simulated datasets, and the results are shown in Section \ref{sec:simulation}. The non-parametric bootstrap algorithm for constructing SCIs for regression coefficients and linear combination of the coefficients is provided in Section \ref{sec:simulation}. In addition, for different correlation structures between the estimated means $\hat \mu(s)$ in the domain $\mathcal{S}$, we demonstrate how conservative the method is when only a finite number of confidence sets are constructed, compared to the SCI nominal coverage rate. We showcase the advantages of our method over the previous approach \cite{sommerfeld2018confidence} in both the simulations and the real data application. Following the simulations, we exhibit two motivating applications in two distinct domain: probability contour for mean temperature difference map for climate change, and logistic regression for determining whether statin is protective against the severe outcome of Coronavirus disease 2019 (COVID) patients in Section \ref{sec:wide_application}. We conclude with a brief discussion in Section \ref{sec:discussion}.

\section{Theory}\label{sec:theory}
\subsection{Setup}\label{sec:notation_setup}
The goal of inverse set estimation is to estimate the set $\mu^{-1}(U)=\{ \s \in \mathcal{S}:\mu(\s) \in U \}$ where $\mu: \mathcal{S} \mapsto \mathbb{R}$ is an unknown deterministic function, $U$ is a fixed subset of $\mathbb{R}$, and $\mathcal{S}$ is a closed indexing set. The "point estimator" of the true inverse set is:
$$
\hat \mu^{-1}(U)=\{ \s \in \mathcal{S}:\hat\mu(\s) \in U \}.
$$

Similar to the point estimate of a scalar parameter, we need a ``lower bound'' and an ``upper bound'' for the estimated inverse set. Therefore, we introduce the data-dependent outer confidence set $\mathrm{CS}_{\mathrm{out}}(U)$ and the data-dependent inner confidence set $\mathrm{CS}_\mathrm{in}(U)$ with the goal that the true inverse set $\mu^{-1}(U)=\{ \s \in \mathcal{S}:\mu(\s) \in U \}$ is ``sandwiched'' within them:
$$
\mathrm{CS}_\mathrm{in}(U) \subseteq \mu^{-1}(U) \subseteq \mathrm{CS}_{\mathrm{out}}(U).
$$

\subsection{Estimating inverse upper excursion sets}
The central idea of this article is that such confidence sets can be obtained by inverting SCIs. Let $\hat B_l(\s)$ and $\hat B_u(\s)$ denote the estimated lower and upper SCI functions at pre-specified level $\alpha$ such that:
$$
\mathbb{P} \left[\forall \s \in \mathcal{S}: \hat B_l(\s) \leq \mu(\s) \leq \hat B_u(\s) \right] = 1- \alpha
$$

Because the function $\mu$ and the SCIs are generally not one-to-one functions, the inversion can get complicated depending on the interval $U$. We simplify this issue by setting $U$ as half of the real line, and this is often the set that researchers are interested in. We can define the following {\em inverse upper excursion set} at level $c$ as:
$$
\mu^{-1}[c,+\infty)=\{\s \in \mathcal{S} \mid \mu(\s) \geq c\}
$$
In addition, we define the following sets as the inner and outer confidence sets for the inverse upper excursion set $\mu^{-1}[c,+\infty)$ for a single level $c$:
$$
\mathrm{CS}_\mathrm{in}[c,+\infty) := \hat B_l^{-1}[c,+\infty)=\left\{s \in \mathcal{S} \mid  \hat B_l(\s) \geq c \right\}
$$

$$
\mathrm{CS}_{\mathrm{out}}[c,+\infty) :=\hat B_u^{-1}[c,+\infty)=\left\{\s \in \mathcal{S} \mid  \hat B_u(\s) \geq c\right\}
$$

In Figure \ref{fig:1d_dense_CS_showcase}, the red horizontal lines are the $\mathrm{CS}_\mathrm{in}[c,+\infty)$, whereas the union of red, green and blue horizontal lines are the $\mathrm{CS}_{\mathrm{out}}[c,+\infty)$. Henceforth, we distinguish between the inference when $c$ is a single level and when the inference is simultaneous over multiple choices of the level $c$.

\subsubsection{Single level confidence set from SCI}
In \cite{hasenstab2015identifying,hasenstab2016robust}, after constructing a bootstrap percentile SCI, the authors claim that the true mean is greater than $c=0$ in the region where the estimated lower interval is greater than $0$. However, no probability or confidence statement is given. This is one of the ad-hoc examples of using SCI for inverse set estimation in applications. The following Proposition \ref{prop:vanilla} provides a formal justification for the procedure above, stating that for a single level $c$, $\hat B_l^{-1}[c,+\infty)$ is a set such that we are at least 95\% confident that $\forall \s \in\hat B_l^{-1}[c,+\infty), \mu(\s)\geq c$.

\begin{proposition}\label{prop:vanilla}
For a fixed level $c \in \mathbb{R}$, and SCIs with $\alpha$ type I family-wiser error rate, we have
   $$
   \mathbb{P}\left(\hat B_l^{-1}[c,+\infty) \subseteq \mu^{-1}[c,+\infty)\right)\geq \mathbb{P}\left(\forall \s \in \mathcal{S}:\hat B_l(\s) \leq \mu(\s) \leq \hat B_u(\s) \big) \right)=1-\alpha
   $$
\end{proposition}
\begin{proof}
  Define the following events 
  $$E := \big\{ \hat B_l^{-1}[c,+\infty) \subseteq \mu^{-1}[c,+\infty) \big\},$$ 
  and 
  $$E_{SCI} := \left\{\forall \s \in \mathcal{S}: \hat B_l(\s) \leq \mu(\s) \leq  \hat B_u(\s) \big) \right\}.$$
  We want to show:
   $$
   E_{SCI} \implies E.
   $$
  Conditioning on the $E_{SCI}$ event, assume for a fixed $\s' \in \mathcal{S}$, we have $\hat B_l(\s') \geq c$, then 
  $$
  \mu(\s')\geq \hat B_l(\s') \geq c
  $$
  by $E_{SCI}$. This means that $\forall \s \in \mathcal{S}: \hat B_l(\s) \geq c$, we must also have $\mu(\s) \geq c$ holds as well, which is equivalent to the statement $\hat B_l^{-1}[c,+\infty) \subseteq \mu^{-1}[c,+\infty)$.
\end{proof}

There are two issues with simply converting the lower band of SCI into inner confidence sets. The first major problem is that this inversion is conservative, as indicated by the coverage probability inequality in Proposition \ref{prop:vanilla}, and we do not know how conservative this construction is. In other words, is it possible to construct multiple confidence sets for different levels $c$ and still maintain the type I error rate? And when will the equality hold? Second, this proposition only provides the inner confidence set as $B_l^{-1}[c,+\infty)$, but gives no outer confidence set. This is of interest because the outer confidence set would capture regions where the signal is not strong, and the region outside the outer confidence set would capture regions where $\mu(\s)<c$ with the desired probability. This is additional information not provided by the inner set or lower SCI. We address these issues together in the next section. 

\subsubsection{Simultaneous confidence sets for multiple inverse upper excursion sets}
Improving on Proposition \ref{prop:vanilla}, we obtain equality in the coverage probability by introducing both upper and lower confidence sets for all levels in $\mathbb{R}$, as shown in  Theorem \ref{theorem:upper}.
\begin{theorem}
\label{theorem:upper}
(Inverse upper excursion set) 
Let $\hat B_l(\s)$ and $\hat B_u(\s)$ be the pre-constructed SCIs on the domain $\mathcal{S}$, then 
\[
\begin{split}
\mathbb{P} \left(\forall c \in \mathbb{R}: \hat B_l^{-1}[c,+\infty) \subseteq \mu^{-1}[c,+\infty) \subseteq \hat B_u^{-1}[c,+\infty) \right) 
= \mathbb{P} \left(\forall \s \in \mathcal{S}: \hat B_l(\s)\leq \mu(\s)  \leq \hat B_u(\s) \right) 
\end{split}
\]
\end{theorem}
\begin{proof}
    See appendix section.
\end{proof}

Theorem \ref{theorem:upper} states that inner and outer confidence sets for all levels in the real line can be constructed based on the SCIs and maintain the same coverage probability as the SCIs. By introducing the outer confidence set in the probability statement in Proposition \ref{prop:vanilla}, the inversion will still be conservative for a \textit{single} level $c$ by Theorem \ref{theorem:upper}. The detailed procedure is provided in Algorithm \ref{algorithm:upper}.

\begin{algorithm}
\caption{Simultaneous confidence sets for multiple inverse upper excursion sets}
\begin{algorithmic}[1]\label{algorithm:upper}
\REQUIRE $\alpha$ type I error rate SCI on the domain $\mathcal{S}$ with lower band $\hat B_l(\s)$ and upper band $\hat B_u(\s)$.
\REQUIRE A discrete subset with $m$ elements in $\mathbb{R}$ : $\mathcal{C} = \{c_1, c_2,...,c_m\}$.
\FOR{$c$ in $\mathcal{C}$} 
    \STATE {
    Construct inner confidence set as $\mathrm{CS}_\mathrm{in}[c,+\infty) := \hat B_l^{-1}[c,+\infty)=\left\{\s \in \mathcal{S} \mid  \hat B_l(\s) \geq c \right\}$, and outer confidence set as $\mathrm{CS}_{\mathrm{out}}[c,+\infty) :=\hat B_u^{-1}[c,+\infty)=\left\{\s \in \mathcal{S} \mid  \hat B_u(\s) \geq c\right\}$.
    }
\ENDFOR
\RETURN {$\mathrm{CS}_\mathrm{in}[c,+\infty)$ and $\mathrm{CS}_{\mathrm{out}}[c,+\infty)$ for all $c \in \mathcal{C}$.} 
\end{algorithmic}
\end{algorithm}

\begin{remark}
From the proof of Theorem \ref{theorem:upper} in the appendix, it can be seen that when $\mathcal{C}$ is a strict subset of $\mathbb{R}$, as in Algorithm 1, the procedure is conservative, that is, it is guaranteed that $\mathbb{P}\left(\forall c \in \mathcal{C}: \mathrm{CS}_\mathrm{in}[c,+\infty) \subseteq \mu^{-1}[c,+\infty) \subseteq \mathrm{CS}_{\mathrm{out}}[c,+\infty)\right) > 1-\alpha$. Only when $\mathcal{C}=\mathbb{R}$ the equality holds.
\end{remark}

\subsection{Estimating inverse lower excursion sets}
We can define the {\em inverse lower excursion set}:
$$
\mu^{-1}(-\infty,c]=\{\s \in \mathcal{S} \mid \mu(\s) \leq c\}= \{\mu^{-1}[c,+\infty)\}^{\complement}
$$
where $\{\mu^{-1}(-\infty,c]\}^{\complement}$ is the the closed complement of the inverse upper excursion set. We could not directly take the complement of the events in Theorem \ref{theorem:upper} to obtain the confidence sets for the inverse lower excursion set since there is an additional closure operation. 

The following sets are defined as the outer and inner confidence sets for the inverse lower excursion set $\mu^{-1}(-\infty,c]$:

\begin{equation}
    \mathrm{CS}_\mathrm{in}(-\infty,c] :=\hat B_u^{-1}(-\infty,c]=\left\{\s \in \mathcal{S} \mid  \hat B_u(\s) \leq c\right\}=\{\hat B_u^{-1}[c,+\infty)\}^{\complement} =  \{\mathrm{CS}_\mathrm{out}[c,+\infty)\}^{\complement}
    \label{eq:cs_in_lower}
\end{equation}

\begin{equation}
    \mathrm{CS}_{\mathrm{out}}(-\infty,c] := \hat B_l^{-1}(-\infty,c]=\left\{\s \in \mathcal{S} \mid  \hat B_l(\s) \leq c \right\}=\{\hat B_l^{-1}[c,+\infty)\}^{\complement}=\{\mathrm{CS}_\mathrm{in}[c,+\infty)\}^{\complement}
    \label{eq:cs_out_lower}
\end{equation}

Lemma \ref{lemma:greater_equal} in the Appendix shows that the two events, in which inverse upper or lower excursion sets for all $c \in \mathbb{R}$ are contained in the corresponding confidence sets, are equivalent. This directly leads to Corollary \ref{corollary:lower} below, which guarantees the coverage probability of the confidence sets for inverse lower excursion set.

\begin{corollary}
\label{corollary:lower}
(Inverse lower excursion set) Let $\hat B_l(\s)$ and $\hat B_u(\s)$ be the pre-constructed SCIs on the domain $\mathcal{S}$, then 
\[
\begin{split}
\mathbb{P} \left(\forall c \in \mathbb{R}: \hat B_u^{-1}(-\infty,c] \subseteq \mu^{-1}(-\infty,c] \subseteq \hat B_l^{-1}(-\infty,c] \right) 
= \mathbb{P} \left(\forall \s \in \mathcal{S}: \hat B_l(\s)\leq \mu(\s)  \leq \hat B_u(\s) \right) 
\end{split}
\]
\end{corollary}
\begin{proof}
    Using Lemma \ref{lemma:greater_equal}, we know that the following two events are equivalent
    $$E_1 = \left\{\forall c \in \mathbb{R}: \hat B_l^{-1}[c,+\infty) \subseteq \mu^{-1}[c,+\infty) \subseteq \hat B_u^{-1}[c,+\infty)\right\}$$
    $$E_2 = \left\{\forall c \in \mathbb{R}: \hat B_u^{-1}(-\infty,c] \subseteq \mu^{-1}(-\infty,c] \subseteq \hat B_l^{-1}(-\infty,c] \right\}$$
    From Theorem \ref{theorem:upper}, the Corollary directly follows.
\end{proof}

\begin{remark}
    Observe that $\mathrm{CS}_\mathrm{in}(-\infty,c]$ is the region where we estimate the true mean is less than or equal to $c$ with a pre-defined probability. From Equation \ref{eq:cs_in_lower}, we know $\{\mathrm{CS}_\mathrm{out}[c,+\infty)\}^{\complement} = \mathrm{CS}_\mathrm{in}(-\infty,c]$. Therefore, instead of using confidence sets for inverse lower excursion set,  $\mathrm{CS}_\mathrm{in}[c,+\infty)$ and $\{\mathrm{CS}_\mathrm{out}[c,+\infty)\}^{\complement}$ for inverse upper excursion set (Theorem \ref{theorem:upper}) will be sufficient in finding the region where the true mean is greater or equal to $c$ and the region where the true mean is less or equal to $c$ respectively.
    
    % The complement of the outer confidence set of inverse lower excursion set for level $c$ is where we are confident that the true mean is greater than $c$, whereas the inner set of inverse upper excursion set for the same $c$ is where we are confident that the true mean is greater or equal than $c$. Researchers should use confidence sets constructed for inverse lower excursion sets when the goal is to find where the true mean is strictly greater than certain levels.
\end{remark}

% The construction of confidence sets for inverse lower excursion sets can be done using Algorithm \ref{algorithm:lower}.

% \begin{algorithm}
% \caption{Simultaneous confidence sets for multiple inverse lower excursion sets}\label{alg:cap}
% \begin{algorithmic}[1]\label{algorithm:lower}
% \REQUIRE $\alpha$ type I error rate SCI on the domain $\mathcal{S}$ with lower band $\hat B_l(\s)$ and upper band $\hat B_u(\s)$.
% \REQUIRE A discrete subset with m elements in $\mathbb{R}$ : $\mathcal{C} = \{c_1, c_2,...,c_m\}$.
% \FOR{$c$ in $\mathcal{C}$} 
%     \STATE {
%     Construct inner confidence set as $\mathrm{CS}_\mathrm{in}(L_c) :=\hat B_u^{-1}(L_c)=\left\{\s \in \mathcal{S} \mid  \hat B_u(\s) \leq c\right\}$, and outer confidence set as $\mathrm{CS}_{\mathrm{out}}(L_c) := \hat B_l^{-1}(L_c)=\left\{\s \in \mathcal{S} \mid  \hat B_l(\s) \leq c \right\}$.
%     }
% \ENDFOR
% \RETURN {$\mathrm{CS}_\mathrm{in}(U_c)$ and $\mathrm{CS}_{\mathrm{in}}(U_c)$ for all $c\in \mathcal{C}$.}
% \end{algorithmic}
% \end{algorithm}

\subsection{Estimating inverse interval sets}
Another similar problem of interest is finding the set where the true mean is within a certain interval. For instance, a clinician may want to find patients with blood pressure that are within a certain healthy interval \cite{symonds1923blood}. By taking the intersection of the upper inverse excursion set and lower inverse excursion set, we obtain the {\em inverse interval set} where $a<b \in \mathbb{R}$:
$$\mu^{-1}[a,b] :=\{\s \in \mathcal{S}: a \leq \mu(\s)\leq b\}=\mu^{-1}[a,\infty) \bigcap \mu^{-1}(-\infty,b] .$$

We define the following inner and outer confidence sets for the inverse interval set $\mu^{-1}[a,b]$: 
$$\mathrm{CS}_\mathrm{in}[a,b] := \hat B_l^{-1}[a,\infty) \bigcap \hat B_u^{-1}(-\infty,b] = \{\s \in \mathcal{S}: \hat B_l(\s) \geq a \wedge \hat B_u(\s) \leq b \}$$
$$\mathrm{CS}_{\mathrm{out}}[a,b]  := \hat B_u^{-1}[a,\infty) \bigcap \hat B_l^{-1}(-\infty,b] = \{\s \in \mathcal{S}: \hat B_u(\s) \geq a \wedge \hat B_l(\s) \leq b \}.$$

To illustrate this in Figure \ref{fig:1d_dense_CS_showcase}, the true inverse interval set $\mu^{-1}[a,b]$, where $a=0.2, b=0.8$, is approximately located at $ \{[0.02, 0.06] \cup [0.08, 0.11] \cup [0.27, 0.35]\}$ on the x-axis. The inner confidence set $\mathrm{CS}_\mathrm{in}[a,b]$ is located in the region on the x-axis by intersecting the red horizontal line at $y = 0.2$ with the uncolored horizontal dashed line at $y = 0.8$. The outer confidence set $\mathrm{CS}_{\mathrm{out}}[a,b]$ is located in the region by intersecting the colored (red, green or blue) horizontal lines at $y = 0.2$ with the complement of the red horizontal line at $y=0.8$. Therefore, the inner confidence set is approximately located at  $ \{[0.03, 0.04] \cup [0.09, 0.10] \cup [0.29, 0.33]\}$ on the x-axis. And the outer confidence set is approximately located at $ \{[0.01, 0.12] \cup [0.25, 0.36]\}$ on the x-axis.

Corollary \ref{corollary: interval} provides the theoretical guarantee for the coverage rate of the inner and outer confidence constructed for inverse interval sets.
\begin{corollary}
\label{corollary: interval}
(Inverse interval set) Let $a<b \in \mathbb{R}$ and $\hat B_l(\s)$ and $\hat B_u(\s)$ be the pre-constructed SCIs on the domain $\mathcal{S}$, then 
\[
\begin{split}
& \quad \mathbb{P} \left(\forall a,b  \in \mathbb{R}, a<b: \mathrm{CS}_\mathrm{in}[a,b] \subseteq \mu^{-1}[a,b] \subseteq \mathrm{CS}_{\mathrm{out}}[a,b]  \right) \\ 
&= \mathbb{P} \left(\forall \s \in \mathcal{S}: \hat B_l(\s)\leq  \mu(\s) \leq \hat B_u(\s) \right).
\end{split}
\]
\end{corollary}
\begin{proof}
    This directly follows from Lemma \ref{lemma:interval_ab} and Theorem \ref{theorem:upper}.
\end{proof}

Corollary \ref{corollary: interval} states that the confidence sets, constructed for all combinations of $a<b \in \mathbb{R}$, are guaranteed to have the same coverage rate as the SCIs. The algorithm for constructing confidence sets for inverse interval sets is shown in Algorithm \ref{algorithm:interval}.

\begin{algorithm}
\caption{Simultaneous confidence sets for multiple inverse interval sets}
\begin{algorithmic}[1]\label{algorithm:interval}
\REQUIRE $\alpha$ type I error rate SCIs on the domain $\mathcal{S}$ with lower band $\hat B_l(\s)$ and upper band $\hat B_u(\s)$.
\REQUIRE $\mathcal{C}=\{(a_i,b_i): i =1,...,m, a_i<b_i \}$.
\FOR{$(a,b)$ in $\mathcal{C}$} 
    \STATE {
    Construct inner confidence set as $\mathrm{CS}_\mathrm{in}[a,b]:= \hat B_l^{-1}[a,\infty) \bigcap \hat B_u^{-1}(-\infty,b]$, and outer confidence set as $\mathrm{CS}_{\mathrm{out}}[a,b]:= \hat B_u^{-1}[a,\infty) \bigcap \hat B_l^{-1}(-\infty,b]$.
    }
\ENDFOR
    \RETURN {$\mathrm{CS}_\mathrm{in}[a,b]$ and $\mathrm{CS}_{\mathrm{out}}[a,b]$ for all $(a,b) \in \mathcal{C}$.}
\end{algorithmic}
\end{algorithm}

Using Corollary \ref{corollary: interval}, the confidence sets produced are guaranteed to have the following probability statement hold with a finite level set $\mathcal{C}=\{(a_i,b_i): i =1,...,m, a_i<b_i \}$:  
$$\mathbb{P}\big(\forall (a,b) \in \mathcal{C}, a<b:\mathrm{CS}_\mathrm{in}[a,b]\subseteq\mu^{-1}([a,b]) \subseteq \mathrm{CS}_{\mathrm{out}}[a,b]\big) > 1-\alpha.$$

\section{Simulations}\label{sec:simulation}
All simulations are based on 5000 Monte Carlo independent realizations. First, we check the validity of our theorem and corollaries, and compare our method with \cite{sommerfeld2018confidence} using continuous 1D and 2D dense functional data. Second, we focus on simulations of constructing confidence sets for inverse upper excursion sets in regression data cases, as these are more common in practice. Simulations for estimating inverse upper excursion sets of the mean function on a 2D grid of predictors using linear and logistic regression are conducted. Third, the discrete domain setting is demonstrated by constructing confidence sets for inverse upper excursion sets of coefficients in linear regression. Concurrently, we examine our non-parametric bootstrap SCI algorithm for linear regression under different covariate counts and sample sizes. We conclude with a comparison of the conservativeness of the coverage probabilities when constructing confidence sets for small number levels.

\subsection{Estimation of excursion sets of dense functional data}
We followed same setup as in \cite{telschow2022simultaneous} and generated functional signal-plus-noise 1D and 2D data:
\begin{align*}
    &(\mathbf{1 D}): Y(s)=\sin (8 \pi s) \exp (-3 s)+\frac{(0.6-s)^{2}+1}{6} \cdot \frac{\mathbf{a}^{T} \mathbf{K}^{A}(s)}{\left\|\mathbf{K}^{A}(s)\right\|}, \quad s \in[0,1]\\
    &(\mathbf{2 D}):Y(\s)=s_{1} s_{2}+\frac{s_{1}+1}{s_{2}^{2}+1} \cdot \frac{\mathbf{b}^{T} \mathbf{K}^{B}(\s)}{\left\|\mathbf{K}^{B}(\s)\right\|}, \quad \s=\left(s_{1}, s_{2}\right) \in[0,1]^{2}
\end{align*}
with $\mathbf{a} \sim \mathcal{N}\left(0, I_{7 \times 7}\right)$ and $\mathbf{b} \sim \mathcal{N}\left(0, I_{36 \times 36}\right)$. The vector $\mathbf{K}^{A}(s)$ has entries $K_{i}^{A}(s)=\left(\begin{array}{l}6 \\ i\end{array}\right) s^{i}(1-s)^{6-i}$, the $(i, 6)$ th Bernstein polynomial for $i=0, \ldots, 6$ and $\mathbf{K}^{B}(\s)$ is the vector of all entries from the $6 \times 6$-matrix $K_{i j}(\s)=\exp \left(-\frac{\left\|\s-\boldsymbol{x}_{i j}\right\|^{2}}{2 h^{2}}\right)$ with $\boldsymbol{x}_{i j}=(i, j) / 6$ and $h=0.06$. Examples of sample paths of the signal-plus-noise models and the error fields can be found in  \cite{telschow2022simultaneous}. Figures \ref{fig:1d_dense_CS_showcase} and \ref{fig:2d_dense_CS_showcase} display the true mean function over the space of support, where the 1D functional data was sampled from the corresponding model on an equidistant grid of 200 points of $[0,1]$, and the 2D functional data was simulated on an equidistant grid with 50 points in each dimension. The SCIs were generated through a multiplier bootstrap procedure whose details can be found in \cite{telschow2022simultaneous} Appendix A. Once we have the SCIs, confidence sets are directly obtained by inversion.

\begin{figure}
    \centering
    \includegraphics[width=14cm]{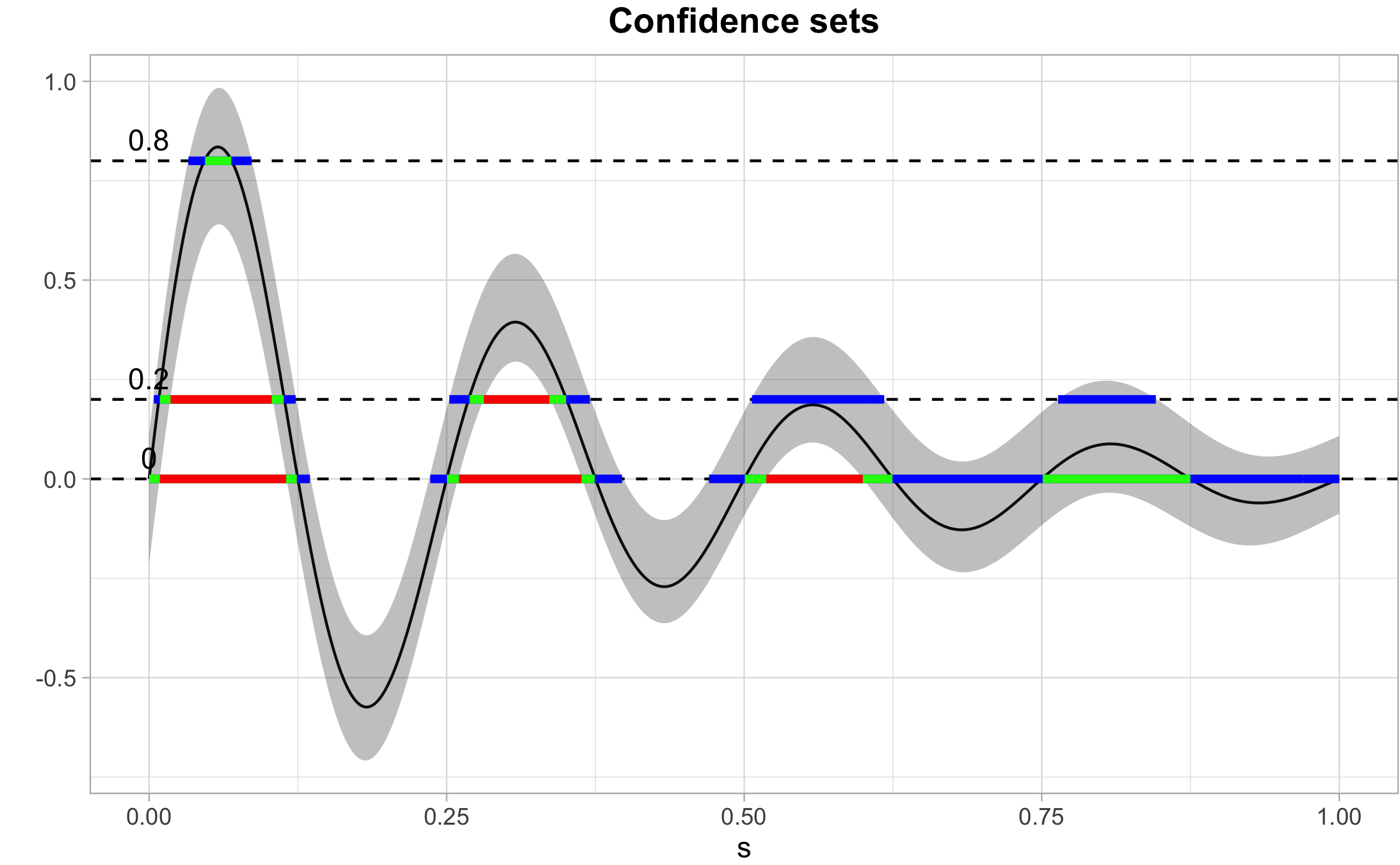}
    \caption{1D dense functional data simulation showcase. Demonstration of using SCB to find regions of $s$ where the true mean is greater than or equal to the three levels $0, 0.2, 0.8$ for 1D dense functional data. The gray shaded area is the 95\% SCB, the solid black line is the true mean. The red horizontal line shows the inner confidence sets (where the lower SCB is greater than the corresponding level) that are contained in the true inverse set represented by the union of the green and red horizontal line (where the true mean is greater than the corresponding levels); the outer confidence sets are the union of the blue, green and red line (where the upper SCB is greater than the corresponding levels) and contain both the true inverse sets and the inner confidence sets.}
    \label{fig:1d_dense_CS_showcase}
\end{figure}

\begin{figure}
    \centering
    \includegraphics[width=14cm]{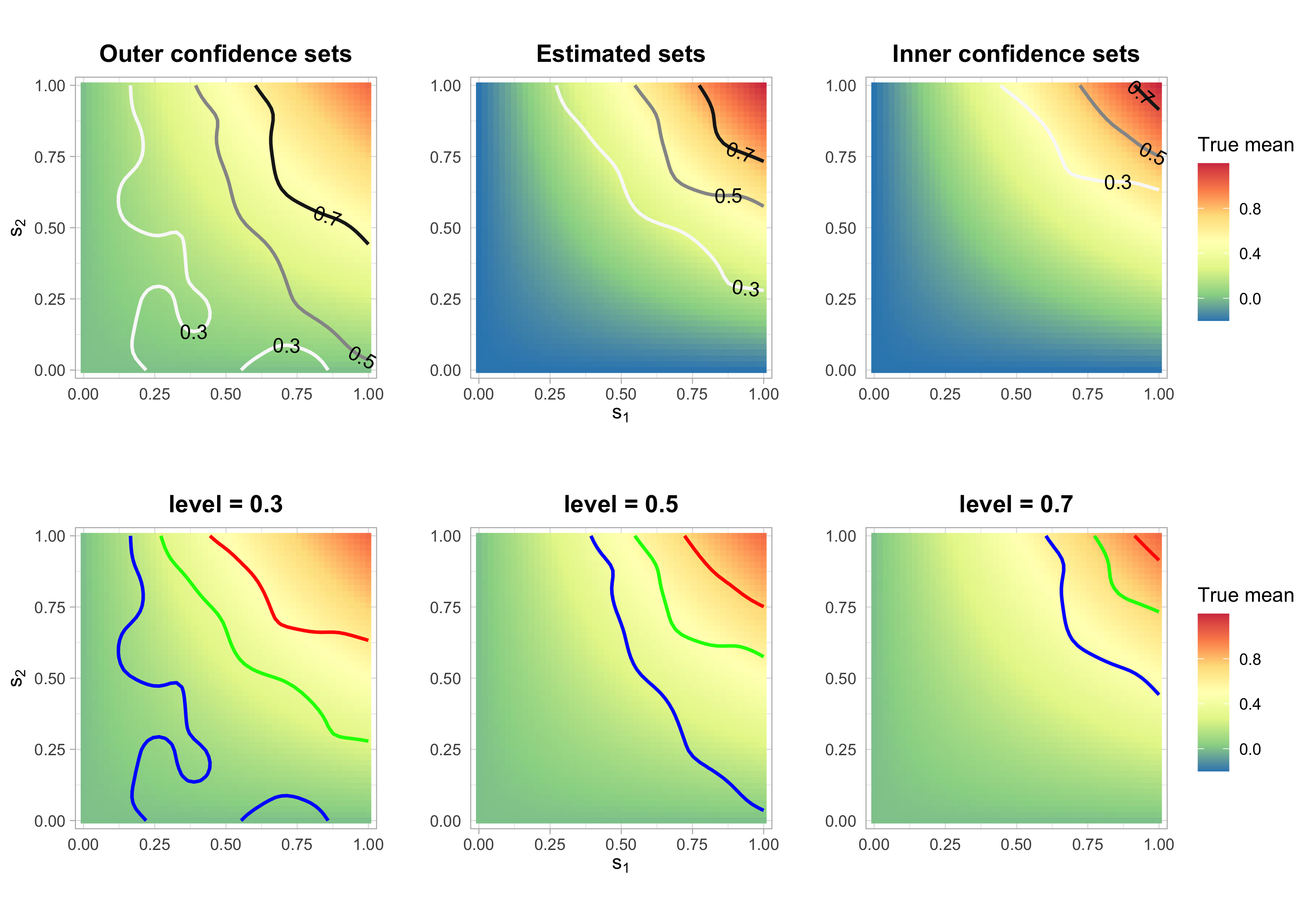}
    \caption{2D dense functional data simulation showcase. The first row displays the contours of the confidence sets in one single plot for the outer confidence sets, estimated inverse set and inner confidence sets, respectively. The three plots in the second row display the contours of the confidence sets for where the true mean is greater or equal to the individual level 0.3, 0.5 or 0.7 respectively. The blue line is the contour of the outer confidence set, the green line is the contour of the estimated inverse set and the red line is the contour of the inner confidence set.}
    \label{fig:2d_dense_CS_showcase}
\end{figure}

Figure \ref{fig:1d_dense_CS_showcase} displays 1D simulated data when the sample size is $N=20$ at each grid point for one realization. Figure \ref{fig:2d_dense_CS_showcase} shows 2D simulated data for one realization of the confidence sets when the sample size is $N=50$. If the researcher wants to find the region where the true mean is greater than or equal to $0.3$, $0.5$ and $0.7$, then the most right plot in the first row of Figure \ref{fig:2d_dense_CS_showcase} will be useful. If the researcher wants to find the region where the true mean is greater or equal to $0.3$ (inside the red contour line) and less than $0.3$ (outside the blue contour line) with 95\% confidence, it would be helpful to investigate the most left panel of the second row of Figure \ref{fig:2d_dense_CS_showcase}.

% The red horizontal line indicates the region where we are at least 95\% confident that all the points inside have a true mean greater or equal to the corresponding level. For example, all the true mean in the region where $0.02 \leq s \leq 0.12$ is guaranteed to be greater than or equal to $-0.1$ with at least 95\% confidence. On the other hand, the dashed horizontal lines without color are where we are at least 95\% confident that all the points inside have a true mean less than the corresponding level. For example, all the true means in the region where $0.14<s<0.22$ are guaranteed to be less than $-0.1$ with at least 95\% confidence. 

In Table \ref{tab:dense_sim_result}, we demonstrate the validity of Theorem \ref{theorem:upper} (labeled as Upper in the table), Corollary \ref{corollary:lower} (Lower) and Corollary \ref{corollary: interval} (Interval). The SCI coverage rate is calculated as the percentage of simulation instances such that the true means for each grid point are contained in the corresponding confidence intervals. The coverage rates for the confidence sets are defined as the percentage of simulation instances such that the inner confidence set is contained in the true inverse set, and the true inverse set is contained in the outer confidence set for every pre-defined level $c$. For inverse upper and lower excursion sets, we checked the containment of confidence sets over 5000 equidistant levels $c$ across the range of minimum and maximum of the true mean function. For the inverse interval set, we checked on the set of intervals $(a,b)$ where $a$ and $b$ are sampled at a step size of $0.005$ ranging from the minimum to the maximum of the true mean, forming a grid with the condition $a<b$. From Table \ref{tab:dense_sim_result}, we can see that the coverage rates of the different types of confidence sets were almost exactly the same as the SCI regardless of the sample size for both 1D and 2D dense functional data with the predefined number of levels, validating our theory. 

\begin{table}
    \caption{Simulation coverage rate for 1D and 2D dense functional data. }\label{tab:dense_sim_result}
    \centering
    \begin{tabular}{ccccccccc} 
        \toprule
         & \multicolumn{4}{c}{1D} & \multicolumn{4}{c}{2D} \\
         \cmidrule(lr){2-5} \cmidrule(lr){6-9}
        {$N$} & SCI & Upper & Lower & Interval  & SCI & Upper & Lower & Interval    \\ \midrule
        10 & 94.98 & 95.00 & 95.00 & 95.18 & 95.16 & 95.16 & 95.18 & 95.32   \\
        20 & 94.94 & 95.00 & 95.00  & 95.22 & 95.96 & 95.96 & 96.06 & 96.42  \\
        30 & 95.26 & 95.34  & 95.34 & 95.54 & 94.92 &  94.96 & 94.96 & 95.40 \\
        50 & 95.02 & 95.10  & 95.10 & 95.40 & 94.92 & 94.98 & 94.98 & 95.42   \\
        100  & 94.94 & 94.98 & 94.98  & 95.58 & 94.88 & 94.90 & 94.96 & 95.46   \\
        150 & 94.58 & 94.64 & 94.64  & 95.06 & 95.10 & 95.12 & 95.18 & 95.64  \\
        \bottomrule
    \end{tabular}
    \\[10pt]
    The simulation standard error is 0.006, calculated as the standard error of a Bernoulli random variable with $p = 0.95$ divided by $\sqrt{5000}$ where 5000 is the number of Monte Carlo simulations.
\end{table}

In addition, we checked the conservativeness of the simultaneous confidence set over a different number of levels for inverse upper excursion sets, and the results are displayed in Figure \ref{fig:dense_sim_result}. As the number of levels decreases, the coverage rate increases, and it depends on the coverage rate of SCI at the corresponding sample size for the realized 5000 simulations. For example, the coverage rate of SCI at sample size 20 for 2D data was 95.96\%, which leads to a higher coverage rate of 96.74\% for 1000 levels and 97.32\% for 50 levels. As the number of confidence sets decreased to 5, the coverage rate rose to around 98\% for both 1D and 2D dense functional data. We also compared our method with the asymptotic single-level confidence sets \cite{sommerfeld2018confidence}, adjusted for multiple levels by using Bonferroni adjustments for every single level. Figure \ref{fig:dense_sim_result} shows that the Bonferroni adjusted multiple single level confidence sets yield coverage lower than the nominal level at small sample sizes and produce higher coverage at larger sample sizes. This is expected since this single-level method is only valid asymptotically for a single level when the sample size is large and becomes conservative when adjusted for multiple levels using Bonferroni adjustment. 

\begin{figure}
    \centering
    \includegraphics[width=14cm]{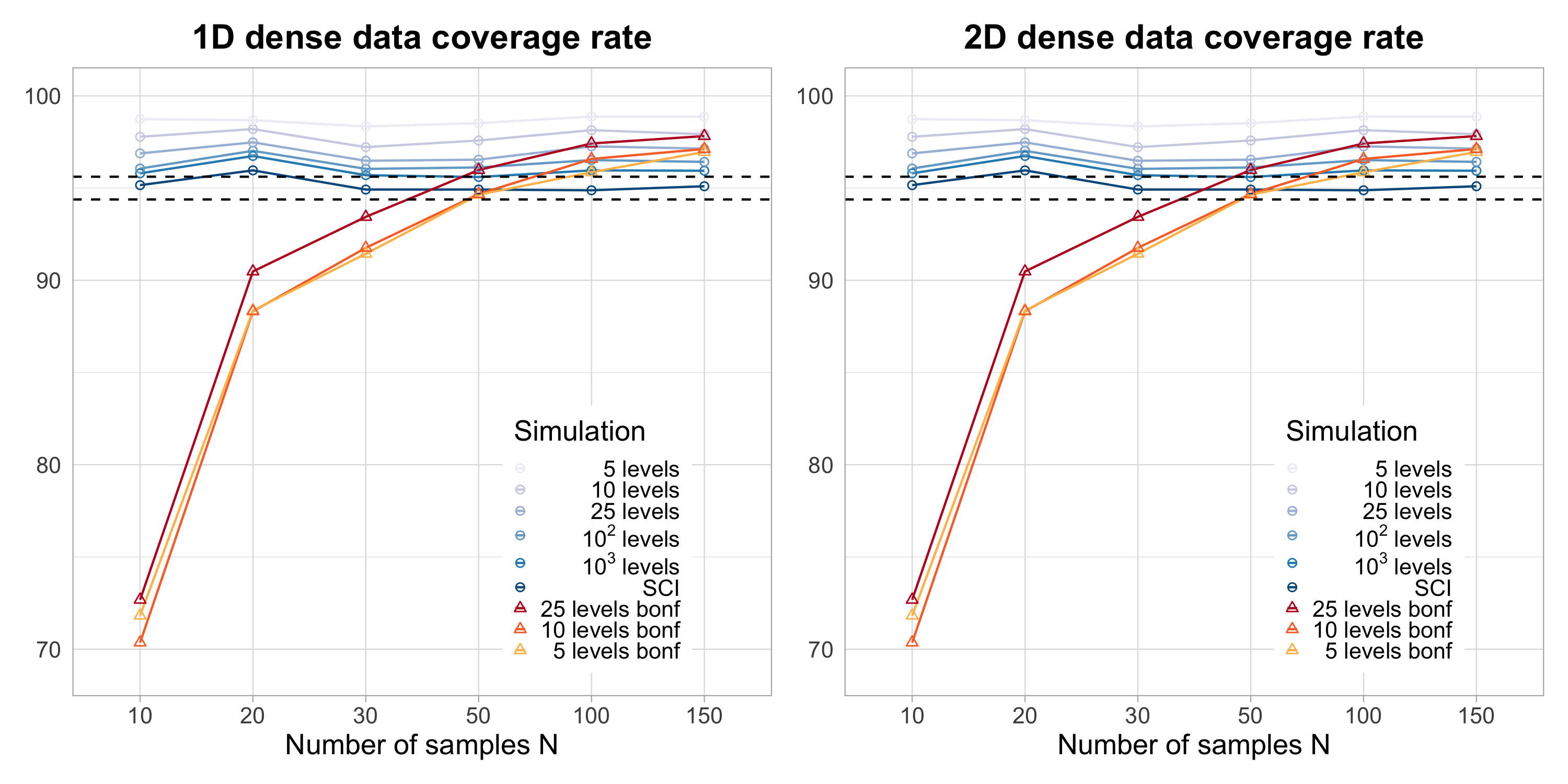}
    \caption{Dense functional data simulation: coverage rate of confidence sets for different number of levels for inverse upper excursion sets. The dashed black line is 95\% plus or minus twice the standard
    error for a Bernoulli random variable with $p = 0.95$ divided by $\sqrt{5000}$.}
    \label{fig:dense_sim_result}
\end{figure}

\subsection{Estimation of excursion sets of regression outcome  }\label{sec:linear_outcome_sim}
We have generated linear and logistic regression data using the following models,
\begin{align}
    &(\mathbf{Linear}): y_i = \beta_0 + \beta_1 x_{i1} + \beta_2x_{i1}^2 + \beta_3x_{i1}^3+\beta_4 x_{i2} + \beta_5 x_{i2}^2 + \beta_6x_{i2}^3 + \epsilon_i \label{linear_eq}\\
    &(\mathbf{Logistic}): log\left(\frac{p_i}{1-p_i} \right) =\beta_0 + \beta_1 x_{i1} + \beta_2x_{i1}^2 + \beta_3x_{i1}^3+\beta_4 x_{i2} + \beta_5 x_{i2}^2 + \beta_6x_{i2}^3\label{logistic_eq},
\end{align}
where $i=1,...,N$ denote the index for subjects and $N$ for training sample size, $\boldsymbol{\beta} = (\beta_0, \beta_1,\beta_2,\beta_3,\beta_4,\beta_5,\beta_6)=(-1,1,0.5,-1.1,-0.5,0.8,-1.1)$, the error $\epsilon_i \iid N(0, 2)$ and $p_i$ denotes the probability of $y_i=1$ for the logistic regression model. The predictors for training the model $\boldsymbol{x}_i = (1, x_{i1},x_{i1}^2,x_{i1}^3, x_{i2}, x_{i2}^2, x_{i2}^3)^T$ are generated from two independent standard normal distributions for $x_{i1}$ and $x_{i2}$, and we denote the design matrix of the training data as $\boldsymbol{X} = (\boldsymbol{x}_1,..., \boldsymbol{x}_n)^T$. The predictions are made on an equidistant grid of 100 points on the interval $[-1,1]$ for the two predictors. We denote by $\Tilde{\boldsymbol{X}}$ the design matrix for the prediction grid, which can also be called the test data design matrix. For our simulation setup, $\Tilde{\boldsymbol{X}}$ is a $10,000$ by $7$ matrix since for each dimension we have 100 points spanning from -1 to 1. Note that the rows of $\Tilde{\boldsymbol{X}}$ are equivalent to $\s$ in the Theory section \ref{sec:theory}.

To generate SCI for the mean outcome on the prediction data grid, we implemented non-parametric bootstrap \cite{freedman1981bootstrapping}. Compared to other bootstrap methods, the non-parametric bootstrap method requires fewer assumptions and is robust under finite sample size setup. We introduce a linear function $f(\boldsymbol{\beta}, \boldsymbol{X})$ that takes in the coefficients vector and design matrix and returns a vector with the same length as the number of rows of $\boldsymbol{X}$. For our simulation setup, this linear functions are the equations \ref{linear_eq} and \ref{logistic_eq} without the errors. The generation process for SCIs of the linear regression mean function is detailed in Algorithm \ref{algorithm:SCI_linear}.

\begin{algorithm}
\caption{SCI for the mean outcome of regression on a fixed test set design matrix}
\begin{algorithmic}[1]\label{algorithm:SCI_linear}
    \REQUIRE Training data outcome $\boldsymbol{y}$ and design matrix $\boldsymbol{X}$. test set design matrix $\Tilde{\boldsymbol{X}}$.
    \REQUIRE Number of Bootstraps $L$ and a fixed function $f(\boldsymbol{\beta}, \boldsymbol{X})$, empty vector $\boldsymbol{r}^{max}$.
    \STATE Obtain $\hat{\boldsymbol{\beta}}$ fitted on the training data $\boldsymbol{y}$ and $\boldsymbol{X}$ using least squares.
    \STATE Calculate the estimated mean outcome vector on the test set design matrix (prediction grid) $\hat{E}(\Tilde{\boldsymbol{y}}):=f(\hat{\boldsymbol{\beta}}, \Tilde{\boldsymbol{X}})$ and the estimated standard deviation $\hat{\boldsymbol{\sigma}}$ of $\hat{E}(\Tilde{\boldsymbol{y}})$. For linear regression, this is the estimated mean outcome on the same scale as $\boldsymbol{y}$, whereas for generalized linear regression, $\hat{E}(\Tilde{\boldsymbol{y}})$ is the estimated linear mean outcome before taking the transformation.
    \FOR{$b$ in $1:L$}
        \STATE Resample with replacement on the training data set: $\boldsymbol{y}_b$, $\boldsymbol{X}_b$.
        \STATE Fit the model on the resampled training data $\boldsymbol{y}_b$, $\boldsymbol{X}_b$, and obtain $\hat{\boldsymbol{\beta}}_b$.
        \STATE Calculate the vector for the estimated mean $\hat{E}(\Tilde{\boldsymbol{y}}_b) := f(\hat{\boldsymbol{\beta}}_b, \Tilde{\boldsymbol{X}})$ and the standard deviation $\hat{\boldsymbol{\sigma}}_{b}$ for every point in the prediction grid.
        \STATE Calculate the of standardized absolute residual vector for every sample in the test set design matrix
        $\boldsymbol{r}_b= \frac{|\hat{E}(\Tilde{\boldsymbol{y}}_b) - \hat{E}(\Tilde{\boldsymbol{y}})|}{\hat{\boldsymbol{\sigma}}_{b}}$. Append the max element of $\boldsymbol{r}_b$ to $\boldsymbol{r}^{max}$.
    \ENDFOR
    \STATE Take the $1-\alpha$ quantile over the empirical distribution generated by $\boldsymbol{r}^{max}$ as our threshold $a$.
    \RETURN SCI on the test set design matrix (prediction grid) as $\left(\hat{E}(\Tilde{\boldsymbol{y}}) - a\times \hat{\boldsymbol{\sigma}}, \hat{E}(\Tilde{\boldsymbol{y}}) + a\times \hat{\boldsymbol{\sigma}}\right)$. For generalized linear regression, transform the SCI back to the data scale using the monotone link function.
\end{algorithmic}
\end{algorithm}

Figures \ref{fig:linear_outcome_showcase} and \ref{fig:logistic_outcome_showcase} display the regression true mean 2D plot with the estimated confidence set contour lines overlayed on top when the training sample size is $N=500$. The outer confidence set, estimated inverse sets and inner confidence sets for all levels are displayed in the first row three plots, and each plot in the second row displays the outer, estimated and inner sets all at once in a single plot for each level. For example, first row and third column of Figure \ref{fig:linear_outcome_showcase} displays the estimated inner confidence sets for level $-1$, $-0.5$ and $0$ simultaneously. For any value of $-1<x_1<1$ and $-1<x_2<-0.75$, the true mean of the linear function conditioned on the two predictors is greater than $0$ with at least 95\% confidence. In Figure \ref{fig:logistic_outcome_showcase} second row and second column, we can see that for any value of $x_1>0$ and $x_2 >0.9$, we have at least 95\% confidence that the true probability of classifying $y=1$ is greater than 50\%, which is where the inner confidence set for level $0.5$ is. We also know that if $-0.8<x_1<0$ and $x_2>-0.4$, then we have at least 95\% confidence that the true probability is less than 50\% of classifying $y=1$, which is the region outside the outer confidence set (blue line). This is a much more intuitive interpretation of the predictors' effect than just interpreting coefficients.

In Figure \ref{fig:regre_outcome_result}, we displayed the coverage rate for different numbers of levels for both the linear and logistic regression. The coverage rate for 5 levels increased only around 1\% for linear regression and around 2\% for logistic regression compared to SCI coverage rate. This is due to the high correlation between the estimated means in the prediction grid, which we address in section \ref{sec:cor}. 

The coverage rate for SCI maintains near the nominal 95\% level even for 100 number of samples for both linear and logistic regression for our setup with 6 covariates, as shown in Figure \ref{fig:regre_outcome_result}. This shows the robustness under finite sample of our non-parametric bootstrap algorithm. We also demonstrated that the constructed SCI works for discrete covariates space as shown in Table \ref{tab:grid_point_linear}.

\begin{figure}
    \centering
    \includegraphics[width=14cm]{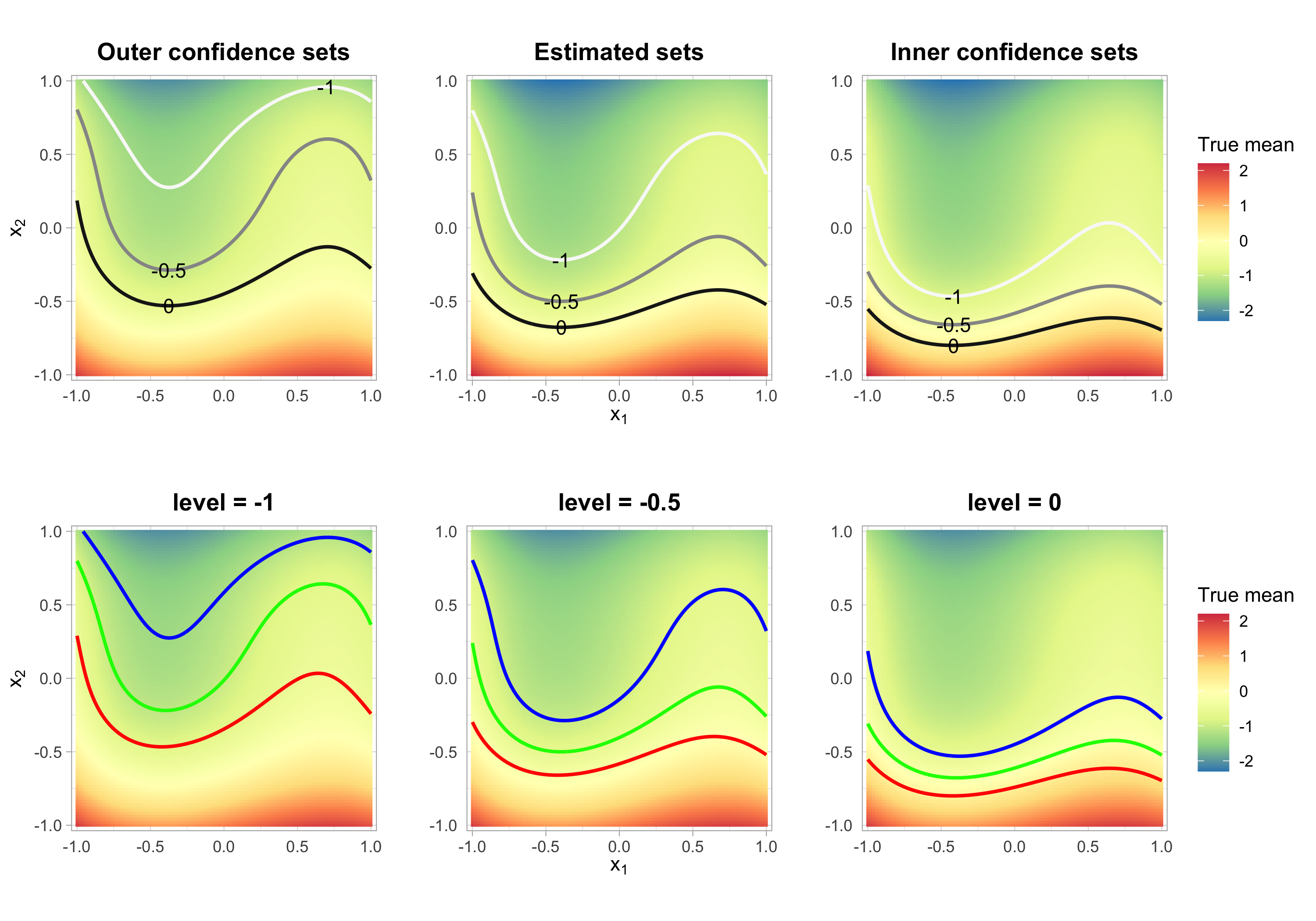}
    \caption{Confidence sets for linear regression predicted mean estimators on a 2D grid. The first row displays the contours of the three levels in one single plot for the outer confidence sets, estimated inverse upper excursion set and the inner confidence sets, respectively. The three plots in the second row display the confidence sets for where the true mean is greater or equal to the individual level -1, -0.5 or 0 respectively. The blue line is the contour of the outer confidence set, the green line is the contour of the estimated inverse set and the red line is the contour of the inner confidence set. The predictions are made on a grid of equidistant grid of 100 points on $[-1,1]$.}
    \label{fig:linear_outcome_showcase}
\end{figure}

\begin{figure}
    \centering
    \includegraphics[width=14cm]{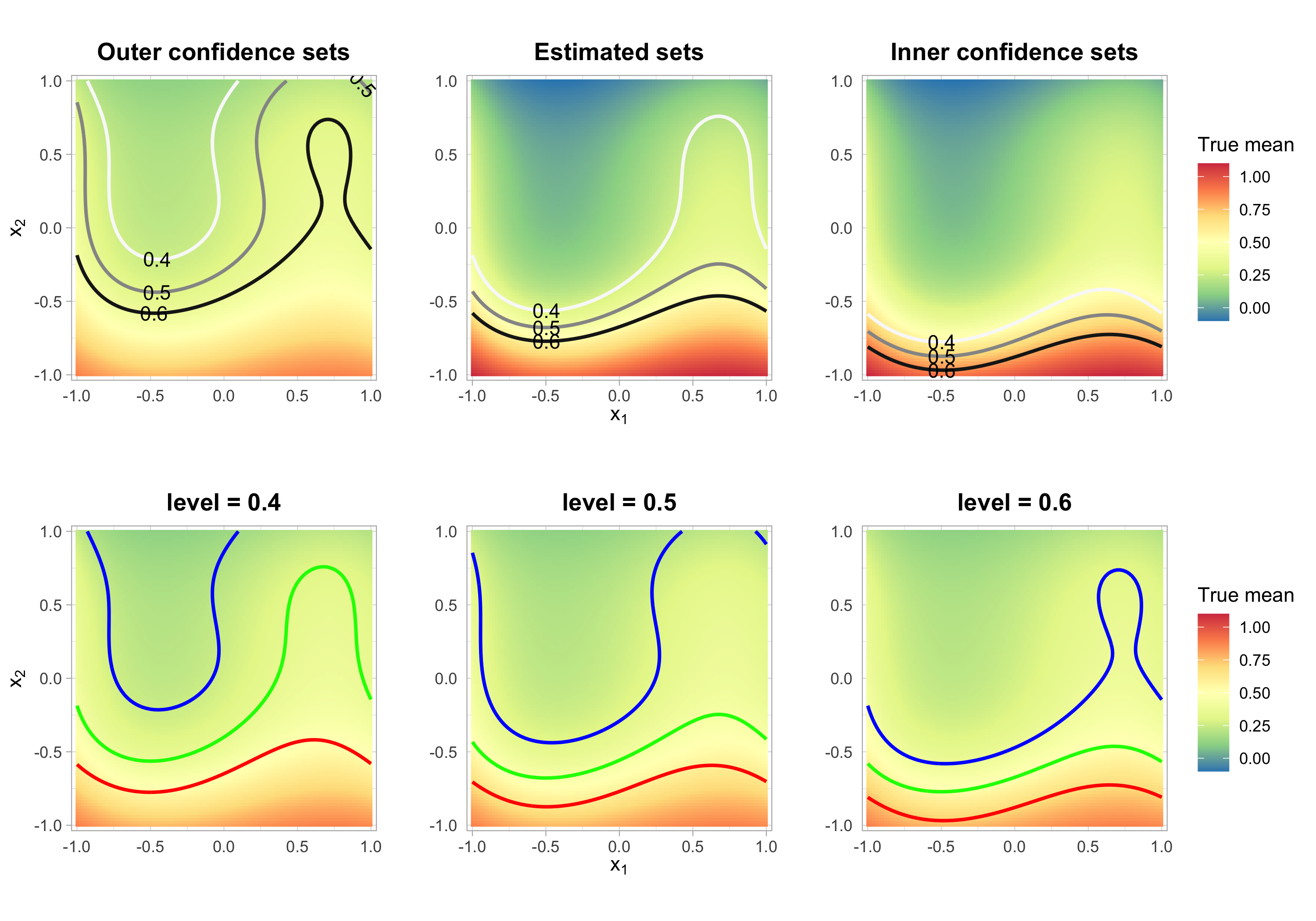}
    \caption{Confidence sets for logistic regression predicted mean estimator for probability on a 2D grid. The first row displays the contours of the three levels in one single plot for the outer confidence sets, estimated inverse upper excursion set and the inner confidence sets, respectively. The three plots in the second row display the confidence sets for where the true mean is greater or equal to the individual level 0.4, 0.5 or 0.6 respectively. The blue line is the contour of the outer confidence set, the green line is the contour of the estimated inverse set and the red line is the contour of the inner confidence set. The predictions are made on a grid of equidistant grid of 100 points on $[-1,1]$.}
    \label{fig:logistic_outcome_showcase}
\end{figure}

\begin{figure}
    \centering
    \includegraphics[width=14cm]{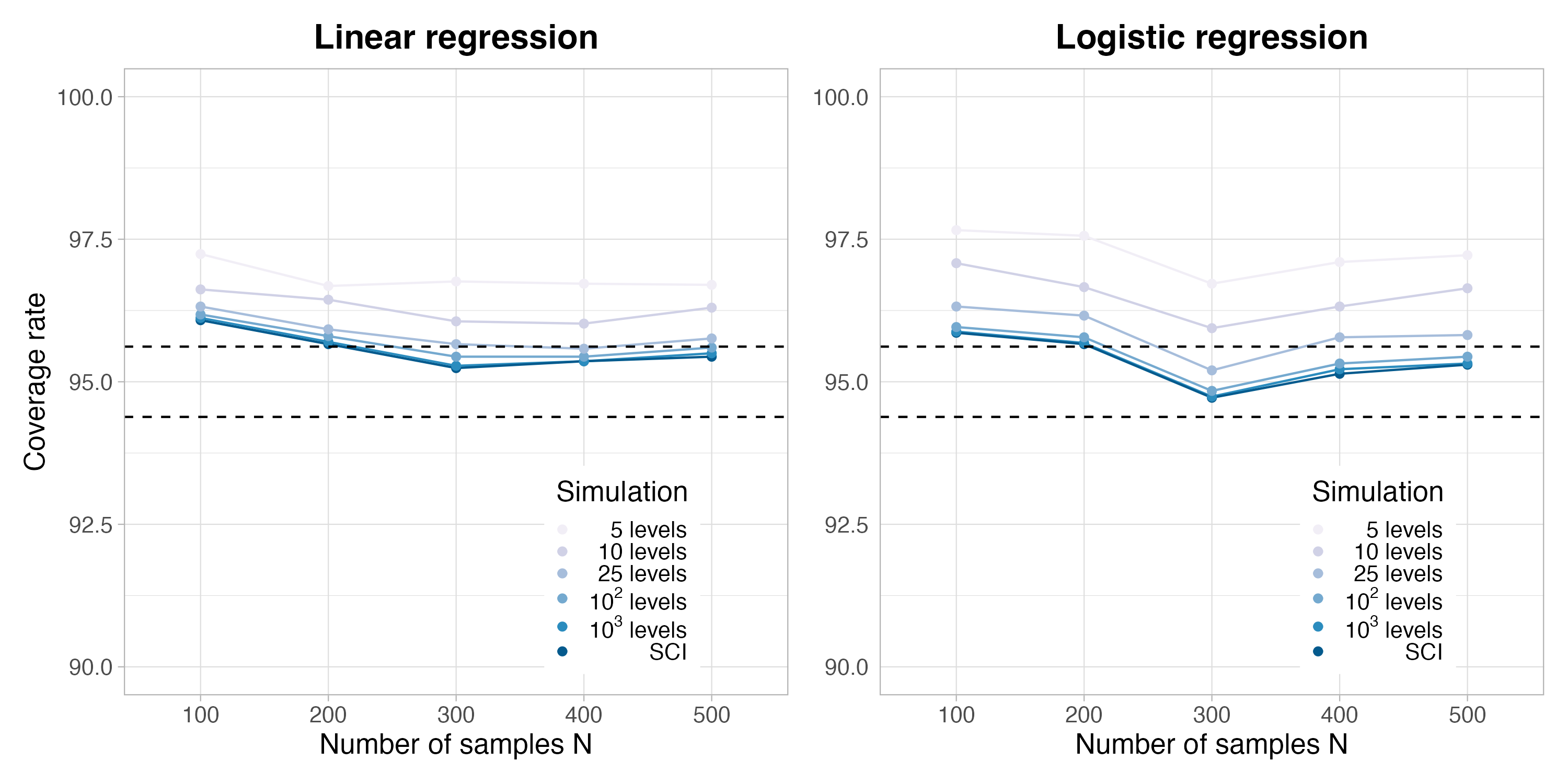}
    \caption{Regression outcome simulation confidence sets coverage rate  for different number of levels for inverse upper excursion sets. The dashed black line is 95\% plus or minus twice the standard
    error for a Bernoulli random variable with $p = 0.95$ divided by $\sqrt{5000}$.}
    \label{fig:regre_outcome_result}
\end{figure}

\subsection{Estimation of excursion sets of regression coefficients}
This simulation for regression coefficients demonstrates the validity of confidence sets on discrete domain by using the SCI for coefficients in linear regression. To flexibly control the number of coefficients, we generated the data under the following model:
$$
y_i = \beta_0 + \sum_{j = 1}^{M-1} \beta_j x_{ij} +\epsilon_i
$$
where $\boldsymbol{\beta}= (\beta_0, \beta_1,...,\beta_{M-1}) \sim N(\boldsymbol{0}, \boldsymbol{I})$ are generated once and fixed for all simulations instances, $\boldsymbol{x}_i=(x_{i1},x_{i2},...,x_{i,M-1}) \sim N(\boldsymbol{0}, \boldsymbol{\Sigma})$ are generated randomly for every new simulation instances, and $\boldsymbol{\Sigma}$ is a auto-regressive covariance matrix with an order of 1, decay factor $\rho=0.4$ and variance of 1 on the diagonal. The irreducible error $\epsilon_i$ follows an independent standard normal. The SCIs of the coefficients were generated similarly to the regression outcome SCI by using non-parametric bootstrapping as shown in Algorithm \ref{algorithm:SCI_linear}. 

Figure \ref{fig:linear_coeff_showcase} shows the confidence sets estimation for the 50 coefficients in one realization of the simulations when $N=500$. The red points are the inner confidence sets for each level which are contained in the true inverse upper excursion sets (green + red points) that are contained in the outer confidence set (blue+green+red points). 

We investigated how the coverage rate changes with both sample size and the number of coefficients in the model for different number of levels. As shown in Figure \ref{fig:regre_coeff_result}, the coverage rates for a small number of levels are more conservative than both the dense functional and regression outcome model and do not vary with either the number of coefficients or the number of samples. The overall conservativeness for the finite number of levels is due to the low correlation between the estimated coefficients as discussed in Section \ref{sec:cor}.

The SCI coverage rate maintained well above the nominal 95\% level, even when the number of sample size is 500 and number of coefficients is 200, as shown in Figure \ref{fig:regre_coeff_result}. This again reinforced the fact that the non-parametric algorithm for linear regression coefficients is robust under finite sample size. 

\begin{figure}
    \centering
    \includegraphics[width=14cm]{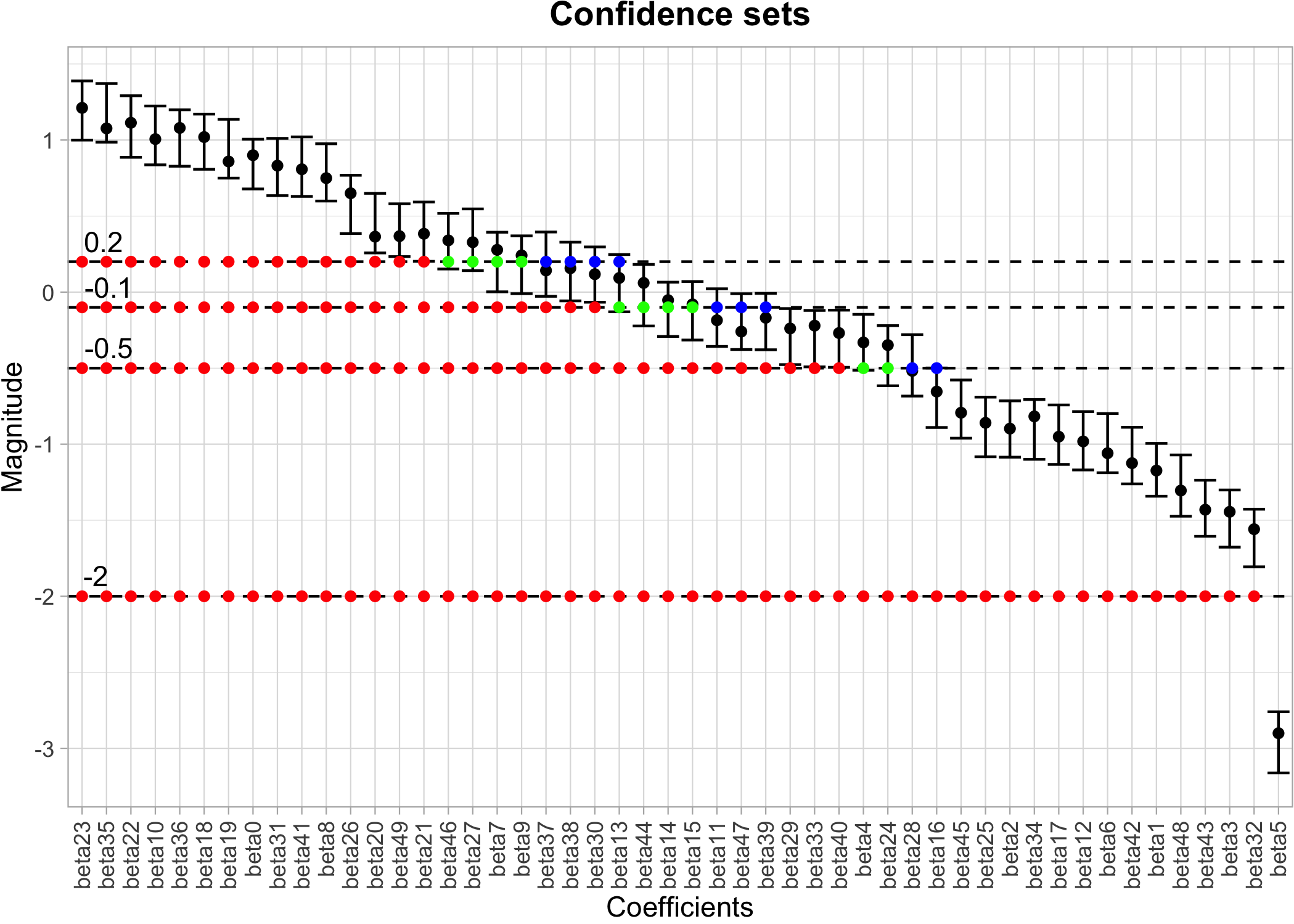}
    \caption{Confidence sets for inverse upper excursion sets of coefficients in the discrete linear regression example. The black points indicate the true means ordered by the lower interval values from left to right. For each level, the red points indicate the inner confidence sets, which are contained in the true inverse upper excursion sets (green+red points) that are contained in the outer confidence set (blue+green+red points).}\label{fig:linear_coeff_showcase}
\end{figure}

\begin{figure}
    \centering
    \includegraphics[width=14cm]{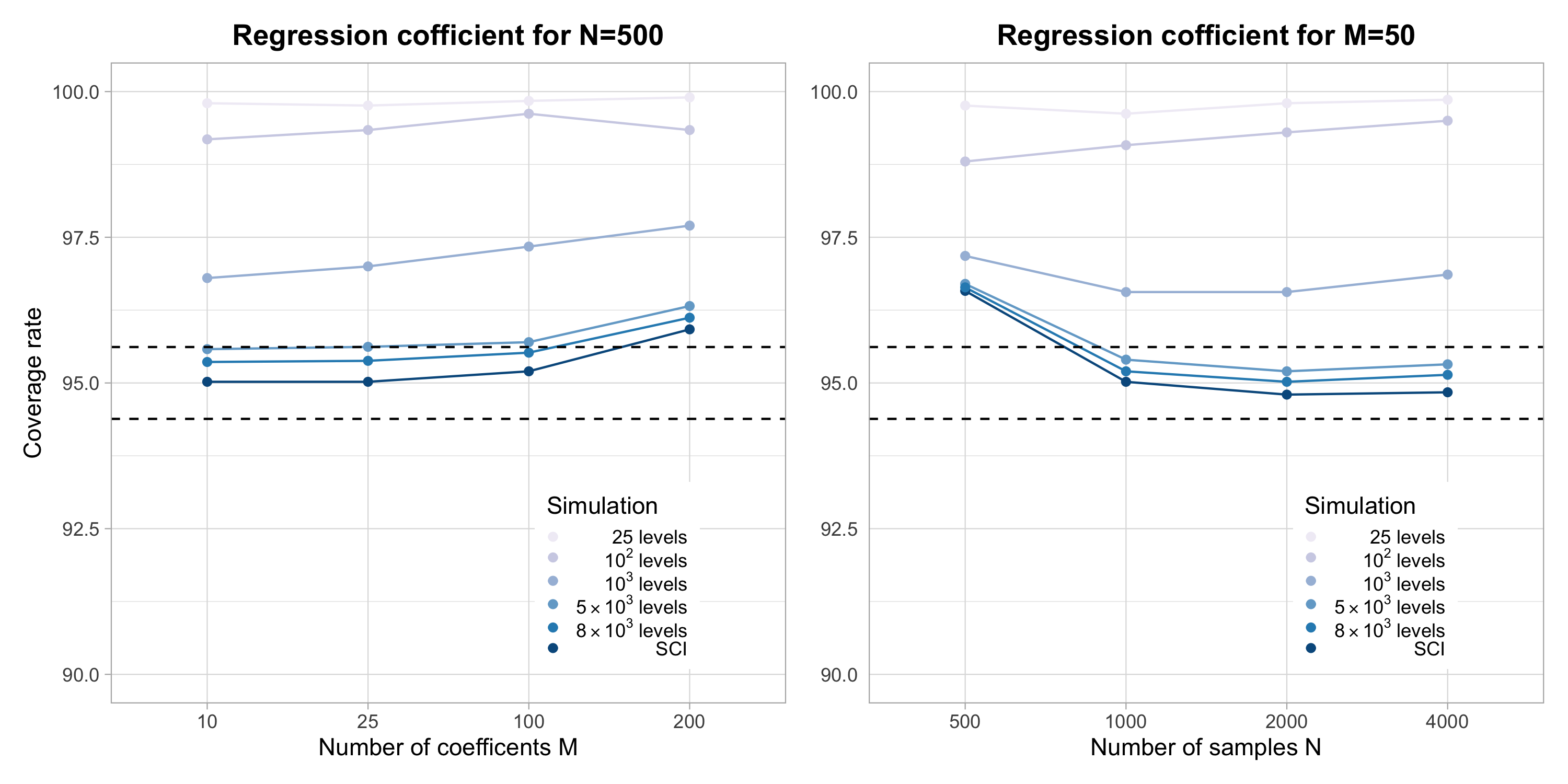}
    \caption{Regression coefficient simulation confidence sets coverage rate for different number of levels for inverse upper excursion sets. The dashed black line is 95\% plus or minus twice the standard
    error for a Bernoulli random variable with $p = 0.95$ divided by $\sqrt{5000}$.}
    \label{fig:regre_coeff_result}
\end{figure}

\subsection{Conservativeness of confidence sets depends on correlations of the estimators}\label{sec:cor}
Once the coverage rate of the CSs is above the nominal level, additional increase of coverage rate decreases the power. Decreased power leads to smaller inner set and larger outer set that indicate a less precise estimation of the true inverse set. Therefore, this section investigates how the coverage probability of confidence sets for a finite number of levels changes with the correlations between the estimator $\hat \mu(s)$ in the domain $\mathcal{S}$. In Figure \ref{fig:corr_diff}, we displayed the pairwise absolute correlation density of the estimators for the simulation setup as shown in Figures \ref{fig:linear_outcome_showcase}, \ref{fig:logistic_outcome_showcase}, and \ref{fig:linear_coeff_showcase} for linear grid prediction means, logistic grid prediction means and discrete coefficients. For example, the green line shows the density for the correlations between every pair of the mean predictions in the grids for linear regression in one simulation instance. We can see most pairs of the discrete coefficient estimates have low absolute correlation at around 0.10, whereas the linear regression grid prediction means estimators' absolute correlations concentrates at around 0.75 leading to the least conservative coverage for finite levels of confidence sets. The logistic regression prediction mean estimators' absolute correlations are less concentrated at higher values and thus they produce more conservative coverage than the linear regression even though both the linear and logistic models are generated with the same underlying linear model.

An enlightening example in the linear regression setting demonstrates that the conservativeness is more dependent on the correlations of the estimators instead of the number of the estimators that confidence sets are constructed for. The correlation matrix of the coefficient estimators is:
$$
\mathrm{cov}(\hat{\boldsymbol{\beta}})= \sigma^2 (\boldsymbol{X'X})^{-1}
$$
where $\sigma^2$ is the error variance and $\boldsymbol{X}$ is the training design matrix. Then, the pairwise correlation between the mean estimator $\hat y_i$ and $\hat y_j$ for the two points $\boldsymbol{x_i}$ and $\boldsymbol{x_j}$ in the testing grid is:
\begin{equation}
    \mathrm{cor}(\hat y_i, \hat y_j)= \frac{\boldsymbol{x_i}(\boldsymbol{X'X})^{-1}\boldsymbol{x_j}'}{\sqrt{\boldsymbol{x_i}(\boldsymbol{X'X})^{-1}\boldsymbol{x_i}'}\sqrt{\boldsymbol{x_j}(\boldsymbol{X'X})^{-1}\boldsymbol{x_j}'}} 
    \label{eq:cor_regression}
\end{equation}

If $\boldsymbol{x_i}$ and $\boldsymbol{x_j}$ are close in Euclidean distance, then the correlation of the prediction mean estimator would be extremely close to 1 regardless of the correlation of the coefficients $cov(\hat{\boldsymbol{\beta}})= \sigma^2 (\boldsymbol{X'X})^{-1}$. This is because the multiplication of the two square roots in the denominator would have similar value as the numerator in Equation \ref{eq:cor_regression}. We constructed linear regression outcome simulations under the same model setup as above but with the step size of 0.02 between the prediction grid points and only vary the number of grid points for a fixed fitting data sample size of 300. For example, for 5 grid prediction points in one dimension, the grid points would be (-0.04,-0.02,0,0.02,0.04). As the number of grid points increases to 100, that would be the same simulation setup in section \ref{sec:linear_outcome_sim}, thus we omit it from this simulation setup. Table \ref{tab:grid_point_linear} illustrates that if the points are closer together (thus high correlated), the coverage rate of the confidence sets for finite levels would be close to the coverage rate for SCIs regardless of the number of estimators we are constructing.

\begin{figure}
    \centering
    \includegraphics[width=14cm]{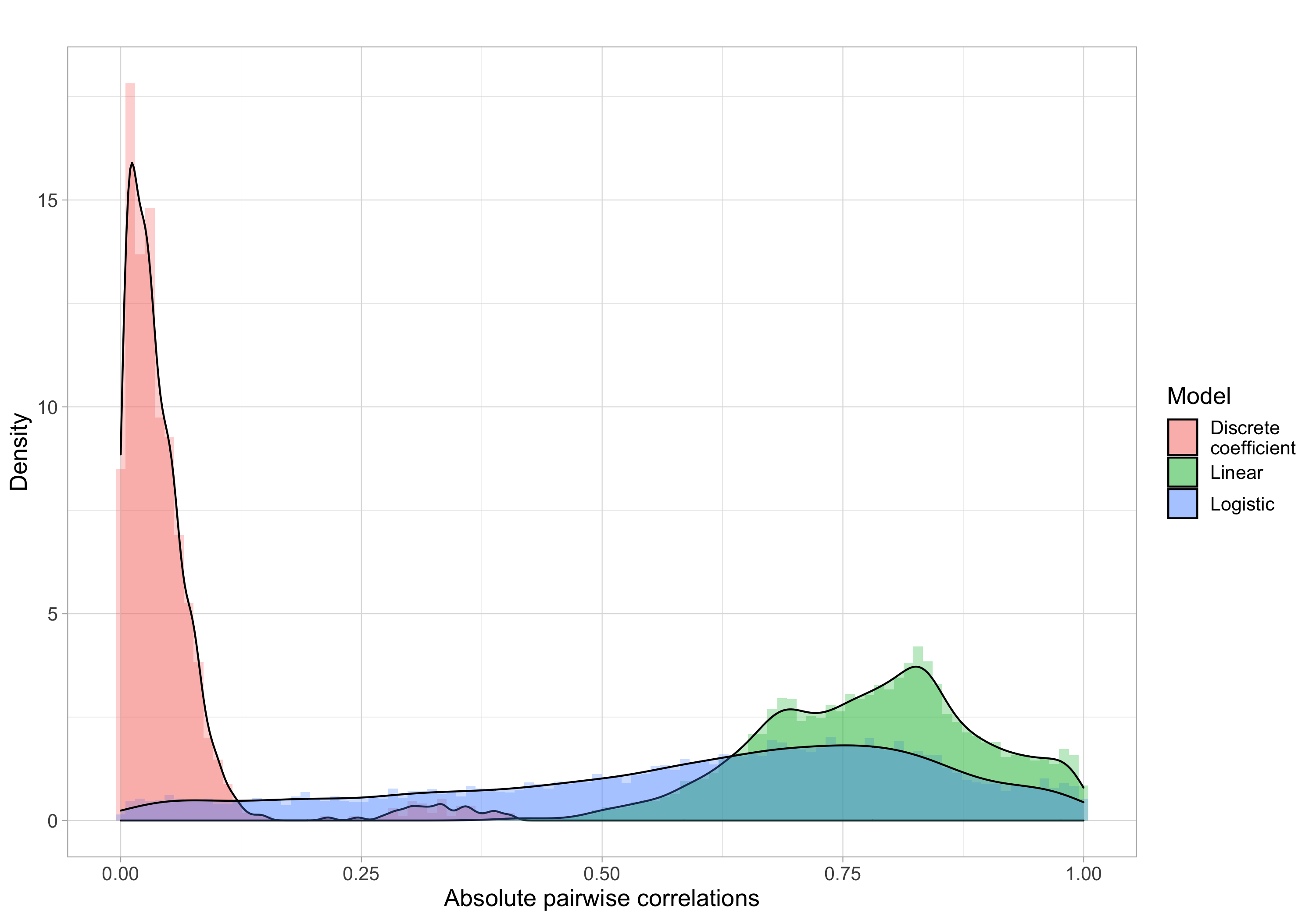}
    \caption{Absolute pairwise correlations of the estimators for different models}
    \label{fig:corr_diff}
\end{figure}

\begin{table}[]
    \centering
    \caption{Simulation coverage rates for fix grid step size for linear regression model set up as shown in Equation \ref{linear_eq}.}
    \begin{tabular}{cccccc} 
        \toprule
        {\# of grid points} & 5 levels & 25 levels & 100 levels & 1000 levels & SCI  \\ \midrule
        5  & 95.70 & 95.48 & 95.26 & 95.16 & 95.16 \\
        10  & 95.28  & 94.88 & 94.80  & 94.66 & 94.62  \\
        20  & 95.96  & 95.38 & 95.20  & 95.08 & 95.06  \\
        50  & 95.82  & 95.14 & 94.98  & 94.94 & 94.94   \\
        80  & 96.50  & 95.40 & 95.36  & 95.20 & 95.18  \\
        \bottomrule
    \end{tabular}
    \\[10pt]
    The simulation standard error is 0.006, calculated as the standard error of Bernoulli random variable with $p = 0.95$ divided by $\sqrt{5000}$ where 5000 is the number of Monte Carlo simulations.
    \label{tab:grid_point_linear}
\end{table}

\section{Applications}\label{sec:wide_application}
\subsection{Excursion set maps for climate data}
With global warming emerging as a serious worldwide environmental issue, it is of interest to assess which geographical regions are particularly at high risk of an increase in temperature. Two sets of 29 spatially registered arrays of mean summer temperatures (June-August) produced by the WRFG climate model as part of the North American Regional Climate Change Assessment Program (NARCCAP) project \cite{mearns2007north,mearns2009regional} are given at a fine grid of fixed locations 0.5 degrees apart in geographic longitude and latitude over North America over two time periods: late-20th century (1971–1999) and mid-21st century (2041-2069). Here we consider the problem of determining which regions have a mean difference in temperature greater than a certain value between the two time periods for summer. This value was set to be $2^\circ C$ in \cite{rogelj2009halfway,anderson2011beyond, sommerfeld2018confidence}, but the value is rather arbitrary. Why would a difference of $1^\circ C$ or even $4^\circ C$ not be of greater importance? The purpose of the analysis presented here is to show how more excursion thresholds can be explored without losing error control.

We followed the same modeling method approach as in \cite{sommerfeld2018confidence}. Briefly, we consider a point-wise linear model with an auto-regressive (AR1) correlation structure for the correlation between different years within every grid point (location):
$$
\begin{array}{ll}
Y_{j}(\s)=T^{(a)}(\s)+m^{(a)}(\s) t_{j}^{(a)}+\epsilon_{j}(\s), & j=1, \ldots, n^{(a)} \\
Y_{j}(\s)=T^{(b)}(\s)+m^{(b)}(\s) t_{j}^{(b)}+\epsilon_{j}(\s), & j=n^{(a)}+1, \ldots, n^{(a)}+n^{(b)}
\end{array}
$$
where $n^{(a)}=n^{(b)}=29$ are the number of years within the "past" and "future" periods. Here $t_j$ is the year number normalized to have mean 0 within each "past" and "future" periods, and the model coefficients are $T(\s)$ and $m(\s)$. Our main interest is to estimate the difference $T^{(b)}(s) -T^{(a)}(s)$. The $90\%$ SCIs were obtained through multiplier bootstrap which accounted the correlations between different locations. For more details, please refer to  \cite{sommerfeld2018confidence}.

In the first row of Figure \ref{fig:climate_result}, we display the point estimate of inverse upper excursion set, outer and inner confidence sets of temperature difference 1, 1.5, 2, 2.5 and 3 Celsius. For example, we are at least $90 \%$ confident that most of the United States and the northern part of Mexico have a summer temperature difference less than $3^\circ C$ as seen in the outer confidence set plot. On the other hand, we are at least $90 \%$ confident that the Rocky Mountains and the Sierra Madre Occidental mountains of Mexico are at risk of exhibiting warming of $2^\circ C$ or more in the given time period which can be seen in the inner confidence set plot.

We also compare our simultaneous finite sample method to the asymptotic single level confidence set in \cite{sommerfeld2018confidence} using the same data at a temperature difference of $2^\circ C$. Taking the ratio of the SCI quantiles (the value $a$ in Algorithm \ref{algorithm:SCI_linear}) calculated from the multiplier bootstrap for the two methods (single level confidence set used a subset of the support for multiplier bootstrap whereas simultaneous confidence used all the points on the support for bootstrap), the simultaneous confidence set threshold is around 127\% (3.8/3) larger than single-level confidence set's threshold value. The position and size of the outer and inner confidence set from the two methods do not differ substantially which can be observed by comparing the level = $2^\circ C$ plot in Figure \ref{fig:climate_result} with Figure 1 in \cite{sommerfeld2018confidence}. This leads to a similar interpretation of the results.

\subsection{Prediction uncertainty quantification for severe outcome of COVID and non-COVID patients}
Coronavirus disease 2019 (COVID-19) has caused significant morbidity and mortality worldwide. Statin medications used by cardiovascular disease (CVD) patients may have a protective effect against severe COVID due to their anti-inflammatory effects \cite{castiglione2020statin}. \cite{daniels2020relation} performed a retrospective single-center study of all patients hospitalized at University of California San Diego Health between February 10, 2020 and June 17, 2020 (n = 170 hospitalized for COVID-19, n = 5,281 COVID-negative controls). Here we will use the same data and a similar multiple logistic regression model to showcase our confidence set method. For more details of the data, please refer to \cite{daniels2020relation}.

The binary outcome of interest is severe outcome of the admitted patient, defined as either admission to the intensive care unit (ICU) or death. The main predictor is whether the patient is taking statin medications or not, and we have adjusted for potential confounders such as angiotensin-converting enzyme (ACE) inhibitors, angiotensin II receptor blockers (ARB), sex, age at diagnosis, chronic kidney disease (CKD), hypertension, cardiovascular disease (CVD), diabetes and obesity. Instead of only investigating the main effect of statin in the COVID positive patients, we pooled the COVID negative and positive population (total sample size of $n=5451$) and added an interaction term between COVID positive indicator and statin medication in the logistic regression model. That is, the log odds of the severe outcome are modeled as a linear function of statin, COVID positive indicator, the interaction between the two, and the remaining confounders. 

The use of statin [adjusted odds ratio (aOR) 0.78, confidence interval (CI) 0.66 to 0.93] is associated with decreased probability of severe outcome, whereas COVID indicator (aOR 4.08, CI 2.82 to 5.95) was associated with an increased probability of severe outcome, and the interaction term for COVID indicator and use of statin was marginally significant (aOR 0.51, 0.25 to 1.05) meaning that statin was more protective in the COVID positive patients. Age at diagnosis was a notable contributor to severe outcome with 2\% (CI 1\% to 3\%) increase in the adjusted odds ratio for one year increase of age at diagnosis. 

In order to visualize and interpret the odds ratios for statin, COVID indicator and age in a meaningful way, we constructed confidence sets for inverse upper excursion sets of three levels of severe outcome probability 0.4, 0.5 and 0.6 on the continuous prediction grid of age spanning from 20 to 100 with a step size of 0.1, and on the discrete prediction grid of COVID positive or not and whether taking statin or not, as shown in Figure \ref{fig:covid_confidence_set}. We fixed other variables at ACE = 0, ARB = 0, sex = Male, CKD = 1, hypertension=1, CVD = 1, diabetes=1, obesity = 1. We use Algorithm \ref{algorithm:SCI_linear} to construct the SCIs for the probability of severe outcome on the prediction grid. Then, the corresponding three confidence sets are built and demonstrated in Figure \ref{fig:covid_confidence_set}. If any patient's characteristic falls within the red solid line, the probability of severe disease will be higher than the corresponding level, whereas if any patient's characteristic falls outside the colored (blue+green+red) solid line, the probability of severe disease will be lower than the corresponding level. These assertions hold for all levels and grid points simultaneously with 95\% confidence, or 5\% probability of making any mistakes, assuming that the model is correct. For example, we are confident that patients whose age is from 40 to 60 with COVID and not taking statin (Figure \ref{fig:covid_confidence_set}, top right, red set) will have a higher than 40\% probability of getting severe outcome, whereas the same aged patients without COVID and taking statin (Figure \ref{fig:covid_confidence_set}, bottom left, outside of blue set) will have less than 40\% probability, conditioned on the other variables at the fixed values.

\section{Discussion}\label{sec:discussion}
We have proposed an innovative approach to construct simultaneous confidence sets of multiple thresholds for quantifying the uncertainty in estimating the inverse set. Previously, inverse set estimation methods have often been limited to data supported on a continuous domain such as dense functional spatial data. We demonstrated that inverse set estimations can be used and provide insightful inference in other continuous or discrete data scenarios such as regression prediction and finding the coefficients with values greater than certain threshold levels. 

In addition, our simultaneous confidence set method solved the dilemma of which threshold level to choose for the inverse set. This is especially important since there is often no universal consensus on which threshold level to apply, even within a well-defined problem. Our construction allows the analyst to choose the excursion level freely, without concern for inflating the error rate. This comes with the drawback that our method is conservative when applied to a small number of excursion levels, which leads to loss of power. However, we empirically demonstrated that this conservativeness depends on the specific model and data. For confidence sets of linear regression predictions, the type I error is very close to the nominal 5\% with only 5 levels in the simulation study. Furthermore, if a slight decrease in power is not of great importance, our simultaneous method outshines the single-level method in three major aspects: our method is not asymptotic and valid for finite samples, produces valid confidence sets for multiple levels, and can be widely applicable to different kinds of data. 

We proposed a non-parametric bootstrap algorithm, supplemented by R code, for constructing SCIs in multiple regression and provide a comprehensive evaluation of its performance demonstrating its robustness under finite sample size and control of Type I error rate. To the best of our knowledge, this is first paper to provide implementable R code and comprehensive evaluation for constructing SCIs for multiple regression using non-parametric bootstrap.

The potential of our method is not fully explored, since it is possible to apply our method to other statistical procedures that output SCIs. Confidence sets can be built by inverting the SCIs without additional assumptions or computational costs for these methods. This will aid inference and interpretation in applications such as survival analysis \cite{mckeague2002simultaneous}, longitudinal data analysis \cite{ma2012simultaneous} and genetic SNP effect analysis \cite{park2007simultaneous, qiu2007sharp, hwang2013empirical}.

\begin{appendix}\label{sec:appendix}
\section{Appendix: Proofs}
\subsection{Theorem \ref{theorem:upper}}
\begin{proof} 
From the definition of inner confidence set, outer confidence set and inverse set, the first equality follows:
\begin{align*}
&\quad \quad \mathbb{P}\Big(\forall c \in \mathbb{R}: \hat B_l^{-1}(U_c) \subseteq \mu^{-1}(U_{c}) \subseteq \hat B_u^{-1}(U_c) \Big) \nonumber \\
&=\mathbb{P}\left(\forall c \in \mathbb{R}: \Big( \big\{ \hat{B}_l(\s)\geq c \big\} \subseteq \big\{ \mu(\s) \geq c \big\} \Big) \wedge \Big( \big\{ \mu(\s) \geq c \big\} \subseteq \big\{ \hat{B}_u(\s) \geq c \big\} \Big) \right) \\
&=\mathbb{P} \left(\forall \s \in \mathcal{S}: \hat B_l(\s)\leq \mu(\s)  \leq \hat B_u(\s) \right).
.\end{align*}

To show the second equality, we need to show that the following two events are equivalent:
$$
E_1 =\Big \{ \forall c \in \mathbb{R}: \Big( \big\{ \hat B_l(\s) \geq c \big\} \subseteq \big\{ \mu(\s) \geq c \big\} \Big) \wedge \Big( \big\{ \mu(\s) \geq c \big\} \subseteq \big\{ \hat B_u(\s) \geq c \big\} \Big) \Big \}
$$
$$
E_2 = \Big \{  \forall \s \in \mathcal{S}:  \mu(\s) \geq  \hat B_l(\s) ~\wedge~ \mu(\s) \leq \hat B_u(\s) \Big \}
$$

First, let's show $E_2$ implies $E_1$. Assume $E_2$ happened, then for any fixed $c \in \mathbb{R}$,  

$\forall \s \in \mathcal{S}$ such that $\hat B_l(\s)\geq c$, then $\mu(\s)\geq \hat B_l(\s)\geq c$, 

$\forall \s \in \mathcal{S}$ such that $\mu(\s)\geq c$, then  $\hat B_u(\s) \geq \mu(\s)\geq c$.\\

Second, let's show $E_1$ implies $E_2$. Proof by contradiction. Assume $E_1$ happened and $\exists \s' \in \mathcal{S}$ s.t. 
$$\hat B_l(\s')> \mu(\s') \vee \hat B_u(\s') < \mu(\s') \text{ ($E_2$ does not hold)}$$
Then, any $c' \in (\mu(\s'),\hat B_l(\s')]$ or $c' \in (\hat B_u(\s'), \mu(\s')]$,
$$\hat B_l(\s') \geq c' \text{ but } \mu(\s')<c'$$ or
$$\mu(\s')\geq c' \text{ but } \hat B_u(\s') < c'$$
Contradiction to $E_1$.
\end{proof}

\subsection{Lemma \ref{lemma:greater_equal}}
\begin{lemma}
\label{lemma:greater_equal}
    $$
    E_1 =\Big \{ \forall c \in \mathbb{R}: \Big( \big\{ \hat B_l(\s) \geq c \big\} \subseteq \big\{ \mu(\s) \geq c \big\} \Big) \wedge \Big( \big\{ \mu(\s) \geq c \big\} \subseteq \big\{ \hat B_u(\s) \geq c \big\} \Big) \Big \}
    $$ 
    is equivalent to 
    $$
    E_1' =\Big \{ \forall c \in \mathbb{R}: \Big( \big\{ \hat B_l(\s) > c \big\} \subseteq \big\{ \mu(\s) > c \big\} \Big) \wedge \Big( \big\{ \mu(\s) > c \big\} \subseteq \big\{ \hat B_u(\s) > c \big\} \Big) \Big \}
    $$
\end{lemma}

\begin{proof}
    Proof by contradiction for both directions.
    
    Assume $E_1$ hold and $E_1'$ does not hold, then
    $    \exists \s' \in \mathcal{S}, \exists c' \in \mathbb{R}, \hat B_l(\s') > c' \geq \mu(\s') 
    $ so there exists $\delta >0$ such that
    $$
    \hat B_l(\s') \geq c' +\delta > \mu(\s') 
    $$
    This is a contradiction to part of $E_1$'s statement $\forall c \in \mathbb{R}, \big\{ \hat B_l(\s) \geq c \big\} \subseteq \big\{ \mu(\s) \geq c \big\}$. Similarly if $ \exists \s' \in \mathcal{S}, \exists c' \in \mathbb{R}, \hat \mu(\s') > c' \geq \hat B_u(\s) $ then there exists $\delta >0$ such that
    $$
    \hat \mu(\s') \geq c' + \delta > \hat B_u(\s) 
    $$
    This is a contradiction to the other part of $E_1$'s statement $\forall c \in \mathbb{R}, \big\{ \mu(\s) \geq c \big\} \subseteq \big\{ \hat B_u(\s) \geq c \big\}$.
    
    Assume $E_1'$ hold and $E_1$ does not hold, then there $\exists \s' \in \mathcal{S}, \exists c' \in \mathbb{R}, \hat B_l(\s') \geq c' >\mu(\s')$ so there exists $\delta > 0$ such that
    $$
    \hat B_l(\s') > c'-\delta \geq \mu(\s')
    $$
    This is a contradiction to part of $E_1'$'s statement $\forall c \in \mathbb{R}: \big\{ \hat B_l(\s) > c \big\} \subseteq \big\{ \mu(\s) > c \big\} $. Similarly if $\exists \s' \in \mathcal{S}, \exists c' \in \mathbb{R}, \mu(\s') \geq c' > \hat B_u(\s')$, there exists $\delta > 0 $ such that
    $$
    \mu(\s') > c'-\delta \geq \hat B_u(\s')
    $$
    This is a contradiction to the other part of $E_1'$'s statement $\forall c \in \mathbb{R}:\big\{ \mu(\s) > c \big\} \subseteq \big\{ \hat B_u(\s) > c \big\}$.
\end{proof}

\subsection{Lemma \ref{lemma:interval_ab}}
\begin{lemma}
\label{lemma:interval_ab}
    $$E_1 = \Big \{ \forall c \in \mathbb{R}: \Big( \big\{ \hat B_l(\s) \geq c \big\} \subseteq \big\{ \mu(\s) \geq c \big\} \Big) \wedge \Big( \big\{ \mu(\s) \geq c \big\} \subseteq \big\{ \hat B_u(\s) \geq c \big\} \Big) \Big \}$$
    is equivalent to 
    \begin{align*}
    E_1''=&\Big \{ \forall a,b\in \mathbb{R}, a<b : \Big( \big\{ \hat B_l(\s) \geq a \wedge \hat B_u(\s) \leq b \big\} \subseteq \big\{ a \leq \mu(\s)\leq b \big\} \Big) \\
    &\wedge \Big( \big\{ a \leq \mu(\s)\leq b \big\} \subseteq \big\{ \hat B_u(\s) \geq a \wedge \hat B_l(\s) \leq b  \big\} \Big) \Big \}
    \end{align*}
\end{lemma}

\begin{proof}
    First, it can be easily seen that $E_1''$ implies $E_1$ if we fix $b = +\infty$. 
    Then we show $E_1$ implies $E_1''$ by contradiction proof:
    Assume $E_1$ hold but $E_1''$ does not hold, then this means
    \begin{align}
    &\Big \{ \exists a<b \in \mathbb{R}: \Big( \big\{ \hat B_l(\s) \geq a \wedge \hat B_u(\s) \leq b \big\} \not\subseteq \big\{ a \leq \mu(\s)\leq b \big\} \Big) \label{eq1:lemma2} \\
    &\quad \vee \Big( \big\{ a \leq \mu(\s)\leq b \big\} \not\subseteq \big\{ \hat B_u(\s) \geq a \wedge \hat B_l(\s) \leq b  \big\} \Big) \Big \}\label{eq2:lemma2}
    \end{align}
    \ref{eq1:lemma2} is equivalent to 
    $$\exists \s'\in \mathcal{S},\exists a'<b'\in \mathbb{R}, \hat B_l(\s')\geq a' \wedge \hat B_U(\s')\leq b'$$
    but $\mu(\s')> b'$ or $\mu(\s')<a'$. If $\mu(\s')<a'$ then $\hat B_l(\s') \geq a' >\mu(\s') \Longleftrightarrow \{\hat B_l(\s) \geq a\} \not\subseteq \{\mu(\s) \geq a \}$, contradiction to $E_1$. If $\mu(\s')>b'$, then $\mu(\s')>b'\geq \hat B_u(\s') \Longleftrightarrow \{\mu(\s) > b\}\not \subseteq \{\hat B_u(\s)>b\}$. By lemma \ref{lemma:greater_equal}, this is a contradiction to $E_1$. Similar argument can be made for \ref{eq2:lemma2}.
\end{proof}
\end{appendix}

\bibliographystyle{unsrt}  
\bibliography{invers_set_estimation}

\end{document}